\documentclass[a4paper,11pt]{article}
\usepackage[utf8]{inputenc}
\usepackage[T1]{fontenc}

\usepackage{libertine} 
\usepackage{inconsolata} 
\usepackage{needspace}

\usepackage{cite}
\usepackage[numbers,square]{natbib}
\usepackage{amsmath}
\usepackage{amssymb}
\usepackage{amsfonts,amsthm}
\usepackage{thmtools,thm-restate}
\usepackage{tabularx,lipsum, environ}
\usepackage{bbm,hyperref}
\hypersetup{
	colorlinks=true,
    linkcolor={red!50!black},
    citecolor={blue!50!black},
    urlcolor={blue!80!black},
	bookmarksopen=true,
	bookmarksnumbered,
	bookmarksopenlevel=2,
	bookmarksdepth=3
}
\usepackage[capitalize]{cleveref}
\usepackage{graphicx,float}
\usepackage{linegoal}
\usepackage{caption}
\usepackage{subcaption}
\usepackage{complexity}
\usepackage{enumerate}
\usepackage{multirow}
\usepackage{cleveref}
\usepackage{enumitem,linegoal}
\usepackage{todonotes}
\usepackage{amsthm,color}
\usepackage{amsmath}
\usepackage{graphics}
\usepackage{fullpage}
\usepackage{amssymb}
\usepackage{graphicx}
\usepackage{lipsum}
\usepackage{comment}
\usepackage{caption,enumerate}

\newtheorem{theorem}{Theorem}[section]
\newtheorem{corollary}[theorem]{Corollary}
\newtheorem{lemma}[theorem]{Lemma}
\newtheorem{proposition}[theorem]{Proposition}
\theoremstyle{remark}

\theoremstyle{definition}

\theoremstyle{plain}
\newtheorem{open problem}{Open problem}
\newtheorem{observation}[theorem]{Observation}
\newtheorem{claim}{Claim}

\newenvironment{claimproof}{\begin{proof}\renewcommand{\qedsymbol}{\claimqed}}{\end{proof}\renewcommand{\qedsymbol}{\plainqed}}
\let\plainqed\qedsymbol

\newcommand\restr[2]{{
		\left.\kern-\nulldelimiterspace 
		#1 
		\littletaller 
		\right|_{#2} 
}}

\newcommand{\littletaller}{\mathchoice{\vphantom{\big|}}{}{}{}}

\usepackage{tikz}
\usetikzlibrary{arrows}
\usetikzlibrary{arrows.meta}
\usetikzlibrary{decorations.pathmorphing}
\usetikzlibrary{decorations.markings}
\usetikzlibrary{calc}
\usetikzlibrary{positioning}
\tikzset{
	circ/.style = {circle,draw,fill,inner sep=1.3pt},
	mcirc/.style = {circle,draw,fill,inner sep=1pt},
	circR/.style = {circle,draw=red,fill=red,text=red,inner sep=1.3pt},
	circb/.style = {circle,draw=blue,fill=blue,text=blue,inner sep=1.1pt},
	circr/.style = {circle,draw=red,fill=red,inner sep=1pt},
	scirc/.style = {circle,draw,fill,inner sep=.8pt},
	invisible/.style = {draw=none,inner sep=0pt,font=\tiny},
	nonedge/.style={decorate,decoration={snake,amplitude=.3mm,segment length=1mm},draw}
}

\usepackage{authblk}

\DeclareMathOperator{\wdt}{\textit{wdt}}
\DeclareMathOperator{\WW}{\textit{W}}

\title{Perfect phylogenies via the Minimum Uncovering Branching problem: efficiently solvable cases}

\author[1]{Narmina Baghirova}
\author[2]{Esther Galby}
\author[3]{Martin Milani{\v c}}
\affil[1]{Department
of Decision Support and Operations Research, University of Fribourg, Switzerland}
\affil[2]{Department of Computer Science and Engineering, Chalmers University of Technology and University of Gothenburg, Sweden}
\affil[3]{FAMNIT and IAM,  University of Primorska, Koper, Slovenia}

\date{}

\begin{document}
\maketitle

\begin{sloppypar}
\begin{abstract}
In this paper, we present new efficiently solvable cases of the Minimum Uncovering Branching problem, an optimization problem with applications in cancer genomics introduced by Hujdurovi\'{c}, Husić, Milani{\v c}, Rizzi, and Tomescu in 2018. 
The problem involves a family of finite sets, and the goal is to map each non-maximal set to exactly one set that contains it, minimizing the sum of uncovered elements across all sets in the family. 
Hujdurovi\'{c} et al.~formulated the problem in terms of branchings of the digraph formed by the proper set inclusion relation on the input sets and studied the problem complexity based on properties of the corresponding partially ordered set,
in particular, with respect to its height and width, defined respectively as the maximum cardinality of a chain and an antichain. 
They showed that the problem is \textsf{APX}-complete for instances of bounded height and that a constant-factor approximation algorithm exists for instances of bounded width, but left the exact complexity for bounded-width instances open. In this paper, we answer this question by proving that the problem is solvable in polynomial time.
We derive this result by examining the structural properties of optimal solutions and reducing the problem to computing maximum matchings in bipartite graphs and maximum weight antichains in partially ordered sets.
We also introduce a new polynomially computable lower bound and identify another condition for polynomial-time solvability.

\bigskip
\noindent{\bf Keywords:} 
perfect phylogeny, branching, partially ordered set, antichain, width
\bigskip

\noindent{\bf MSC (2020):}  
05C85, 
05C20, 
05C90, 
06A07, 
92D10
\end{abstract}
\end{sloppypar}



\section{Introduction}
\label{introduction}
In this paper, we investigate the so-called \textsc{Minimum Uncovering Branching} (MUB) optimization problem, introduced by Hujdurovi\'{c}, Husi\'c, Milani\v{c}, Rizzi, and Tomescu in \cite{PPVB}.
The input to the problem is a binary matrix $M$ and the feasible solutions are \emph{branchings} (that is, in-forests) in the so-called \emph{containment digraph} $D_M$, a directed acyclic graph (DAG) derived from the input matrix $M$. The MUB problem consists in computing a branching that minimizes the number of so-called uncovered pairs.
(For the sake of readability, we postpone the precise definitions to \Cref{sec:preliminaries}.)
The problem is equivalent to the \textsc{Minimum Conflict-free Row Split} (MCRS) problem studied by Hujdurovi\'{c}, Ka{\v c}ar, Milani{\v c}, Ries, and Tomescu in~\cite{CompAndAlg}, which is a follow-up work to the one of Hajirasouliha and Raphael~\cite{HajRap}.

\subsection*{Motivation: computational phylogenetics}

The MUB problem has applications in the area of cancer genomics, as it relates to the reconstruction of evolutionary histories of tumors.

One of the main assumptions of the clonal theory of cancer is that cancer arises in a tumor after certain mutations accumulate over time.
For better diagnosis and more targeted therapies (see Newburger et al.~\cite{BetterDiag}), it is essential to understand what mutations and alterations in specific genes cause cancer (see, e.g., Nowell~\cite{Nowell}).
The most accurate way to examine a tumor sample is by sampling and sequencing every single cell contained in the sample, which is impossible with current biotechnological methods. 
Another possible solution for overcoming the challenge of reconstructing the evolutionary history of mutations is to make use of a computational approach.
In the field of phylogenetics, an important approach to studying an evolutionary history is given by the notion of a \textit{perfect phylogeny} (which we formally define in \Cref{sec:preliminaries}). 

The perfect phylogeny evolutionary model can be used for reconstructing the evolutionary history of a tumor, since the tumor progression is assumed to satisfy the following two properties of phylogenetic evolution:
\begin{enumerate}
    \item All mutations in the parent cells are passed to the descendants.
    \item A mutation does not occur twice at the same particular site.
\end{enumerate}

In \cite{HajRap}, Hajirasouliha and Raphael proposed the following computational approach.
We are given a collection of samples of the corresponding tumor, and the task is to reconstruct the history of the evolutionary process of the tumor so that the resulting model corresponds to a perfect phylogeny.  
Given the samples and information about the occurrence of the mutations in each sample, we encode the information into a binary matrix $M$, where rows correspond to samples and columns to mutations. 
The $(i,j)$-th entry of matrix $M$ equals $1$ if the mutation $j$ occurs in sample $i$; $0$ otherwise.

Given a binary matrix $M$, the \emph{perfect phylogeny problem} asks whether the matrix corresponds to a perfect phylogeny evolutionary model.
It is a known result that the rows of a binary matrix $M$ correspond to perfect phylogeny if and only if the matrix is conflict-free, where two columns of the matrix are said to be \emph{in conflict} if there exist three rows of the matrix such that when restricted to the two columns, their entries are $(1,1)$, $(1,0)$, and $(0,1)$, respectively (see \cite{conflictfree,conflictfree2} for more details).
Gusfield (see~\cite{PPfromMlintime}) developed a linear-time algorithm that tests whether a given binary matrix is conflict-free and, if this is the case, also constructs a perfect phylogeny.
In practice, however, each tumor sample is a mixture of reads\footnote{In the field of bioinformatics, in DNA sequencing, the term \emph{read} refers to a sequence obtained at the end of the sequencing process and represents a part of the sequence corresponding to the entire genome.} from several tumor types, and thus the corresponding binary matrix $M$ is not conflict-free.

The MCRS problem, mentioned earlier, aims at explaining each row in a subset of the rows of the binary conflict matrix $M$ with a set of rows, in an overall simplest possible way, so that the resulting binary matrix $M'$ is conflict-free and thus corresponds to a perfect phylogeny. 
Informally, the problem says the following: split each row of a given binary matrix $M$ into the bitwise OR of a set of rows so that the resulting matrix corresponds to a perfect phylogeny and has the minimum possible number of rows among all matrices with this property. Intuitively, we want the minimum possible number of rows in the resulting matrix, since when we split rows that contain information about samples and mutations, each split can be seen as ``altering'' the data encoded in the matrix. Thus, minimizing the number of rows, informally speaking, means preserving the original data as much as the perfect phylogeny constraint allows.

Similarly, in the MUB problem, the goal is to find an optimal branching in the derived containment digraph.
Such a branching can then be used to solve the MCRS problem and reconstruct a perfect phylogeny (see~\cite[Theorem 2.1]{PPVB}).

\subsection*{State of the art}
In~\cite{CompAndAlg}, Hujdurovi\'{c} et al.~showed that the MUB problem is \textsf{NP}-hard, identified a polynomially solvable special case, and gave an efficient heuristic algorithm.
Hujdurovi\'{c} et al.~in \cite{PPVB} strengthened the \textsf{NP}-hardness result to an \textsf{APX}-hardness and gave two approximation algorithms for the MUB problem, described in terms of height and width of the binary matrix $M$ -- where these quantities are defined in terms of parameters of the corresponding containment digraph $D_M$. 
In fact, the authors showed a stronger result, namely that the problem is \textsf{APX}-hard already for instances of height $2$. 
Moreover, they proved a min-max result strengthening Dilworth's theorem, a classical result in the theory of partially ordered sets stating that in any finite partially ordered set, the minimum number of chains in a partition of the set equals the maximum size of an antichain (see~\cite{dilworth}).
Using the min-max result, they improved the heuristic algorithm for the MUB problem from~\cite{CompAndAlg} and derived a polynomial-time algorithm to solve a restricted version of the problem.
Using an Integer Linear Programming formulation of the problem, Husi\'{c} et al.~introduced in~\cite{MIPUP} the Minimum Perfect Unmixed Phylogenies (MIPUP) method for reconstructing a tumor evolution.
The method was tested against four well-known tools and shown to be the most accurate (see~\cite{MIPUP} for more details).
Furthermore, Sheu and Wang~\cite{paramComp} gave an \textsf{FPT} (fixed-parameter tractable) algorithm for the MUB problem, parameterized by the number of uncovered pairs with respect to an optimal branching.
(For definitions regarding parameterized algorithms and complexity, we refer to~\cite{MR3380745}.)

\subsection*{Our results}
The \textsf{APX}-hardness of the MUB problem motivates the search for conditions under which the problem becomes solvable in polynomial time.
In particular, while the \textsf{APX}-hardness proof for the MUB problem (see~\cite{PPVB}) constructs instances of height $2$, the complexity of the MUB problem on instances of bounded width has been left open by previous work.
As our main result, we answer this question by showing that for any integer $k\ge 1$, the MUB problem is polynomial-time solvable on instances of width at most $k$.  
This is a significant improvement over the $k$-approximation algorithm for such instances given by Hujdurovi\'{c} et al.~in \cite{PPVB}.
We also introduce a new polynomially computable lower bound for the problem and use it to show that the MUB problem is polynomial-time solvable on instances where the width of the binary matrix $M$ equals to the number of sinks in the containment digraph $D_M$.
Our results are incomparable with the polynomial special case of the (equivalent) MCRS problem identified by Hujdurovi\'{c} et al.~in~\cite{CompAndAlg}, which corresponds to the case when the out-neighborhood of each vertex of the containment digraph has width at most one. 

Our approach to develop a polynomial-time algorithm for the MUB problem on instances of bounded width is based on several ingredients.
First, we establish a property of the set of optimal solutions that implies that the problem can be reduced to the case when there exists an optimal branching whose set of leaves forms a unique maximum antichain in the containment digraph.
Such an antichain can be found in polynomial time by solving an instance of the \textsc{Maximum Weight Antichain} problem in the containment digraph (which in turn reduces to a maximum flow problem), and a solution of the reduced branching problem can be extended to a solution of the original one by invoking a polynomial-time algorithm from~\cite{PPVB} for a restricted version of the problem.
Finally, the main part of the algorithm performs an exhaustive search of certain highly structured parts of an optimal solution, extending each partial solution to a full solution by solving an instance of the bipartite matching problem.

\subsection*{Structure of the paper}
In~\Cref{sec:preliminaries}, we summarize the main definitions and formally introduce the MUB problem.
In \Cref{sec:bound}, we derive a polynomially computable lower bound on the optimal value of the problem.
In \Cref{sec:polsolcases}, we identify two new polynomially solvable cases: one based on the lower bound and our main result, the case of instances of bounded width.
We conclude the paper with a few open problems in  \Cref{sec:conclusion}.

\section{Preliminaries}
\label{sec:preliminaries}

For a positive integer $n$, we denote by $[n]$ the set $\{1,\ldots, n\}$.
All the graphs in this paper are finite but may be directed or undirected.
Given two vertices $u$ and $v$ in a digraph $D$, we say that $v$ is \emph{reachable} from $u$ if $D$ contains a $u,v$-path; otherwise, we say that $v$ is \emph{unreachable} from $u$. 
A \emph{chain} in $D$ is a set $C$ of vertices of $D$ such that, for every two vertices $u,v\in C$, either $v$ is reachable from $u$ or vice versa.
A \emph{chain partition} of a digraph $D$ is a family $P =\{C_1,\ldots, C_p\}$ of vertex-disjoint chains such that every vertex of $D$ is contained in exactly one chain of $P$.
Given a directed acyclic graph (DAG) $D$, an \emph{antichain} in $D$ is a set of vertices that are pairwise unreachable from each other, and the \emph{width} of $D$, denoted $\wdt(D)$, is the maximum cardinality of an antichain. 
For a vertex $v$ in a digraph $D$, we denote by $d_D^-(v)$ and $d_D^+(v)$ the in- and the out-degree of $v$ in $D$, respectively.
A vertex $v\in V$ with $d^-(v) = 0$ is called a \emph{source}, and a vertex $v\in V$ with $d^+(v) = 0$ is called a \emph{sink}.  
For a set $S \subseteq V(D)$, we denote by $D[S]$ the directed subgraph of $D$ induced by vertices in $S$.
A digraph $D=(V,A)$ is \emph{transitive} if for any three vertices $v_1,v_2,v_3$ if arcs $(v_1,v_2),(v_2,v_3) \in A$ then $(v_1,v_3) \in A$. A \emph{partially ordered set (poset)} is a pair $(S,\preceq)$ such that $S$ is a set and $\preceq$ is a reflexive, antisymmetric, and transitive binary relation on $S$; a \emph{strict partially ordered set (strict poset)} is a pair $(S,\prec)$ such that $S$ is a set and $\prec$ is an asymmetric and transitive binary relation on $S$.
In other words, a strict poset is a transitive DAG, and adding a directed loop at each vertex of a transitive DAG yields a poset.

A \emph{comparability graph} is a graph that admits a \emph{transitive orientation}, that is, an assignment of directions to the edges of the graph such that the resulting directed graph is transitive. 

For a DAG $D$ equipped with a vertex weight function $w:V(D)\to \mathbb{Z}_+$ and a set $X\subseteq V(G)$, we denote by $w(X)$ the weight of $X$ with respect to $w$, that is, the quantity $\sum_{x\in X}w(x)$.
The \textsc{Maximum Weight Antichain} problem takes as input a DAG $D$ and a vertex weight function $w:V(D)\to \mathbb{Z}_+$, and the task is to compute an antichain in $D$ with maximum weight with respect to $w$.
As shown by M\"ohring~\cite{Mohring}, the \textsc{Maximum Weight Antichain} problem can be reduced to a minimum flow problem in a derived network, which in turn can be reduced to maximum flow (see~\cite{MR2839932,MR159700}).
Recently, van den Brand et al.~in \cite{maxflowlinear} introduced an almost linear algorithm that computes exact maximum flows and minimum-cost flows on digraphs with $m$ arcs in time $\mathcal{O}(m^{1+o(1)})$.
Combining these results, we obtain the following.

\begin{theorem}\label{thm:max-weight-antichain}
\textsc{Maximum Weight Antichain} can be solved in time $\mathcal{O}(m^{1+o(1)})$, where $m$ is the number of arcs of the input DAG.
\end{theorem}

A \emph{rooted tree} is a tree with a distinguished vertex, called the \emph{root} of the tree. 
In computational phylogenetics, a \emph{perfect phylogeny} is a rooted tree representing the evolutionary history of a set of $m$ objects such that:
\begin{enumerate}
	\item The $m$ objects bijectively label the leaves of the tree, representing different samples.
	\item There are $n$ binary characters, each labeling exactly one edge of the tree, and representing mutations that occurred during the evolution of a tumor.
\end{enumerate}
A perfect phylogeny with $m$ objects and $n$ binary characters naturally corresponds to an $m \times n$ binary matrix having objects as rows and characters as its columns.
See \Cref{fig:PPmatrix} for an example of a perfect phylogeny with $m=5$ objects and $n=6$ binary characters representing mutations and the corresponding binary matrix. 

\begin{figure}[h!]
	\centering
	\includegraphics[width=0.7\linewidth]{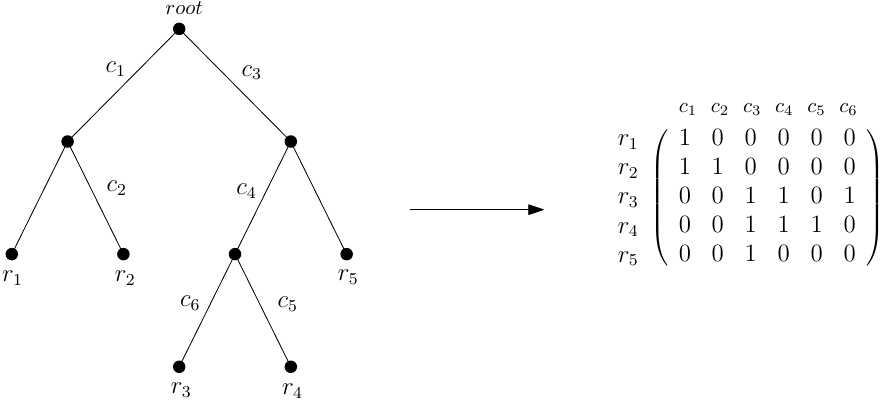}
	\caption{An example of a perfect phylogeny and the corresponding binary matrix $M$. The figure is from \cite{Baghirova}.}
	\label{fig:PPmatrix}
\end{figure}

Recall, following the discussion in the introduction, that we are interested in the opposite direction, not only for binary matrices arising from perfect phylogenies as in \Cref{fig:PPmatrix} but for arbitrary binary matrices. 
More specifically, given a binary matrix containing the information regarding samples and mutations that occurred in those samples, we would like to reconstruct a perfect phylogeny that \emph{best} explains the history of somatic mutations in the corresponding tumor. 
To formalize this notion Hajirasouliha and Raphael introduced in~\cite{HajRap} the so-called \emph{Minimum Split Row} optimization problem, and Hujdurovi\'{c} et al.~in \cite{CompAndAlg} introduced a variant of the problem called \textsc{Minimum Conflict-free Row Split} (MCRS). 
Instead of considering the MCRS problem, we study in this paper a problem equivalent to MCRS, the so-called \emph{Minimum Uncovering Branching} (MUB) problem, introduced by Hujdurovi\'{c} et al.~in~\cite{PPVB}.
We omit the formal definition of the MCRS problem and refer the reader to~\cite{CompAndAlg,PPVB,Baghirova} for details.

\subsection*{The Minimum Uncovering Branching (MUB) problem}
\label{subsec:prelim2}

Let $M\in \{0,1\}^{m \times n}$ be a binary matrix. 
We denote the sets of its columns and rows by $C_M$ and $R_M$, respectively. 
Given a column $c_j \in C_M$, the \emph{support set} of $c_j$ is the set defined as follows: $\textit{supp}_M(c_j) = \{r_i \in R_M \colon M_{i,j}=1\}$, where $M_{i,j}$ represents $(i,j)$-entry of $M$. 
The \emph{containment digraph} $D_M = (V,A)$ corresponding to $M$ has vertex set $V = \{\textit{supp}_M(c)\colon c \in C_M\}$ and an arc set $A = \{(v_i,v_j) \colon v_i, v_j \in V\text{ and }v_i \subset v_j\}$.  
See \cite[Figure 10]{Baghirova} for an example.

Note that the containment digraph corresponding to any binary matrix $M$ is a DAG. 
The containment digraph of a binary matrix $M\in \{0,1\}^{m \times n}$ can be computed in time $\mathcal{O}(n^2m)$, directly from the definition.
Note that the vertices of the containment digraph $D_M$ correspond to the family of support sets of the columns of $M$, and the arcs correspond to the relation of proper inclusion of those sets. 
Hence, the containment digraph represents a (strict) poset $P_M = (V(D_M), \subset)$.
Since we will not consider non-strict partial orders in this paper, we will refer to a strict partially ordered set simply as a \emph{poset}.
For the poset $P_M = (V(D_M), \subset)$ corresponding to the containment digraph $D_M$, an element $m\in P_M$ is \emph{maximal} if $m$ is a sink in $D_M$, that is, there is no element of $P_M$ properly containing $m$.
For a binary matrix $M$, we denote by $\wdt(M)$ the width of the corresponding containment digraph $D_M$, that is, the maximum number of vertices in an antichain in $D_M$. 
For a vertex weight function $w:V(D_M)\to \mathbb{Z}_+$ and a set $X\subseteq V(D_M)$, we denote by $w(X)$ the \emph{weight} of $X$ with respect to $w$, that is, the quantity $\sum_{x\in X}w(x)$.
Moreover, let $\alpha_w(M) := \text{max}\{w(X) \colon X$ is an antichain in $D_M\}$.

Let $M$ be a binary matrix.
A \emph{branching} $B$ of $D_M=(V,A)$ is a subset $B$ of $A$ containing at most one outgoing arc from each vertex. 
Given a vertex $v\in V$ and a set $F\subseteq A$ of arcs of $D_M$, we denote by $d^-_{F}(v)$ and $d^+_{F}(v)$ the in- and the out-degree of $v$ in the digraph $(V,F)$, respectively.
A \emph{leaf} of $B$ is a vertex $u \in V$ such that $d^-_B(u) = 0$. 
For a vertex $v \in V$, we define $N^-_B(v) := \{u\in V \colon (u,v) \in B\}$. 
Given a branching $B$ of $D_M$, a row $r \in R_M$, and a vertex $v \in V$ such that $r \in v$, we say that $r$ is \emph{covered in $v$ with respect to branching $B$} if $r$ is contained in some $u \in \bigcup N^-_B(v)$, that is, there exists some $u\in N^-_B(v)$ such that $r\in u$. 
Otherwise, we say that $r$ is \emph{uncovered} in $v$ with respect to $B$. 
A \emph{$B$-uncovered pair} is a pair $(r,v)$ such that $r\in v$, $r \in R_M$, $v \in V(D_M)$, and $r$ is uncovered in $v$ with respect to $B$.
We denote by $U(B)$ the set of all $B$-uncovered pairs. 
See \cite[Figure 11]{Baghirova} for an example of a branching $B$ and the set of the corresponding uncovered pairs.

Following~\cite{PPVB}, we denote the minimum number of uncovered pairs over all branchings $B$ of $D_M$ by $\beta(M)$. 
We refer to it as the \emph{uncovering number} of $M$. 
We call a branching $B$ with the minimum number of uncovered pairs an \emph{optimal branching}.
The corresponding optimization problem is formalized as follows.

\noindent\rule[0.5ex]{\linewidth}{1pt}
{\sc Minimum Uncovering Branching (MUB):}\\
\noindent\rule[0.5ex]{\linewidth}{0.5pt}
{\em Input:}\quad A binary matrix $M$.\\
{\em Task:} \quad Compute the uncovering number $\beta(M)$.\\
\vspace{-2mm}
\noindent\rule[0.5ex]{\linewidth}{1pt}\\

Let us remark that, except for possibly identical columns of the matrix $M$, this problem is essentially the same as the problem described in the abstract: Given a family $\mathcal{F}$ of finite sets, find a mapping that assigns each non-maximal set in the family to exactly one of the sets containing it, in a way that minimizes the sum, across all sets $X$ in the family, of the elements in $X$ that are not contained in any of the sets mapped to $X$.
Indeed, the family $\mathcal{F}$ of finite sets is the family of the support sets of the columns of $M$ (that is, the vertex set of $D_M$), mappings assigning each non-maximal set in the family to exactly one of the sets containing it correspond bijectively to arc-maximal branchings of $D_M$, and the sum, across all sets $X$ in the family, of the elements in $X$ that are not contained in any of the sets mapped to $X$, is exactly the number $U(B)$ of uncovered pairs for the corresponding branching $B$ of $D_M$.

In \cite{PPVB}, Hujdurovi\'{c} et al.~introduced the concept of linear branching. 
Let $M$ be a binary matrix and $D_M=(V,A)$ be its corresponding containment digraph. 
A branching of $D_M$ is said to be \emph{linear} if it contains at most one incoming arc into each vertex. 
For later use, let us note the following.

\smallskip
\begin{observation}\label{obs:chain-partitions-and-linear-branchings}
To any chain partition $P =\{C_1,\ldots, C_p\}$ of a containment digraph $D_M$, there corresponds a linear branching $B^P$ of $D_M$ consisting of all arcs of the form $(v^i_{j},v^i_{j+1})$ such that $i\in \{1,\ldots, p\}$ and the vertices of the chain $C_i$ are ordered as $\{v_1^i,\ldots,v_{k_i}^i\}$ where $k_i = |C_i|$ so that \hbox{$v^i_1\subset \cdots \subset v_{k_i}^i$}.
Moreover, this correspondence is bijective: for any linear branching $B$ there is a unique chain partition $P$ of $D_M$ such that $B^P= B$.
\end{observation}

Given a binary matrix $M$, we denote by $\beta_{\ell}(M)$ its \emph{linear uncovering number}, defined as the minimum number of elements in $U(B)$ over all linear branchings $B$ of $D_M$. 
Then, the \textsc{Minimum Uncovering Linear Branching (MULB)} problem is defined as a problem of, given a binary matrix $M$, computing its linear uncovering number.
As an immediate consequence of the definitions, we obtain the following.

\smallskip
\begin{observation}\label{obs:MULB-MUB}
Every binary matrix $M$ satisfies $\beta(M) \le \beta_{\ell}(M)$.
\end{observation}
\smallskip

In contrast with the fact that the MUB problem is \textsf{NP}-hard, the MULB problem is solvable in polynomial time (see~\cite{PPVB}).
Thus, the linear uncovering number is a polynomially computable upper bound on the uncovering number.

\section{Lower bounds for the uncovering number}\label{sec:bound}

We now focus on polynomially computable lower bounds on the uncovering number $\beta(M)$, revisiting a known bound and establishing a new one. 
In \Cref{sec:polsolcases}, we will use the latter result to identify an efficiently solvable case.

\begin{sloppypar}
First, we revisit a lower bound due to Hajirasouliha and Raphael~\cite{HajRap}.
Let $M$ be a binary matrix and $r \in R_M$.
The \emph{conflict graph} $G_{M,r}$ is a graph corresponding to matrix $M$ and row $r$ with the following vertex set. 
We associate a vertex with each column of $M$ whose entry in row $r$ equals $1$, and two vertices in $G_{M,r}$ are connected by an edge if and only if the corresponding columns in $M$ are in conflict. 
We say, that $c_i,c_j\in C_M$ are in conflict if there exist three rows $r_p,r_q,r_s \in R_M$ such that the corresponding \( 3 \times 2 \) submatrix of \( M \) has the following form:
\[
M[(r_p,r_q,r_s), (c_i, c_j)] = 
\begin{pmatrix}
1 & 1 \\
1 & 0 \\
0 & 1
\end{pmatrix}.
\]
Given a graph $G$, we denote by $\chi(G)$ its \emph{chromatic number}, that is, the smallest positive integer $k$ such that there exists a mapping from the vertex set of $G$ to a set of $k$ colors such that no two adjacent vertices receive the same color.
The uncovering number of $M$ is bounded from below in terms of chromatic numbers of the graphs $G_{M,r}$ as follows.
\end{sloppypar}

\smallskip
\begin{theorem}[Hajirasouliha and Raphael~\cite{HajRap}]\label{thm:HRLB}
Every binary matrix $M$ satisfies 
\[\beta (M) \ge \sum_{r\in R_M}\chi (G_{M,r})\,.\]
\end{theorem}

The lower bound in \Cref{thm:HRLB} can be computed in polynomial time. This is because the graphs $G_{M,r}$ are perfect. 
Indeed, Hujdurovi\'{c} et al.~\cite{CompAndAlg} provided a complete characterization of conflict graphs (that is, graphs of the form $G_{M,r}$, called \emph{row-conflict graphs} therein), proving that a graph is a conflict graph if and only if its complement is a comparability graph.
Both comparability graphs and their complements are subclasses of the well-known class of perfect graphs~\cite{GmrPerfect}.
Since the chromatic number of a perfect graph can be computed in polynomial time, as shown by Grötschel, Lovász, and Schrijver (see~\cite{Grotschel}), it follows that the lower bound in \Cref{thm:HRLB} is computable in polynomial time.

The lower bound given by \Cref{thm:HRLB} can be expressed more directly in terms of subgraphs of the containment digraph $D_M$.
For a row $r$ of $M$, the \emph{principal subgraph of $D_M$ corresponding to $r$} is denoted by $D_{M,r}$ and defined as the subgraph of $D_M$ induced by the set of vertices $v\in D_M$ such that $r\in v$. 
A \emph{principal subgraph} of $D_M$ is any subgraph of the form $D_{M,r}$. 
Recall that the width of a digraph is the maximum size of an antichain in it.
We express the lower bound given by \Cref{thm:HRLB} in terms of subgraphs of the containment digraph $D_M$ by showing that for every binary matrix $M$ and row $r$ of $M$, the chromatic number of the conflict graph $G_{M,r}$ equals the width of the principal subgraph $D_{M,r}$. 
In the proof of this result, we will also need the notion of the \emph{complement} of an undirected graph $G$, which is the graph $\overline{G}$ with the same vertex set as $G$, in which two distinct vertices are adjacent if and only if they are not adjacent in $G$.

\smallskip
\begin{lemma}
	\label{lemma:principalsubgr}
	For every binary matrix $M$ and row $r$ of $M$, we have $\chi(G_{M,r}) = \wdt(D_{M,r})$.
\end{lemma}

\begin{proof}
Let $M$ be a binary matrix and $r \in R_M$ an arbitrary row. 
Each vertex of $G_{M,r}$ corresponds to some column of $M$.
Note that if two column vectors are identical, then the corresponding vertices are nonadjacent and have the same neighborhood.
Let $\widehat{G}_{M,r}$ be the induced subgraph of $G_{M,r}$ obtained by keeping only one vertex from each maximal set of vertices corresponding to identical column vectors.
Then the chromatic numbers of $G_{M,r}$ and $\widehat{G}_{M,r}$ coincide.
On the one hand, we have $\chi(G_{M,r}) \ge \chi(\widehat{G}_{M,r})$ since $\widehat{G}_{M,r}$ is a subgraph of $G_{M,r}$.
On the other hand, $\chi(G_{M,r}) \le \chi(\widehat{G}_{M,r})$, since a $k$-coloring of $\widehat{G}_{M,r}$ can easily be extended to a $k$-coloring of $G_{M,r}$ by coloring each vertex of  $G_{M,r}$ with the color of the corresponding vertex of $\widehat{G}_{M,r}$.

Let $H_{M,r}$ be the underlying undirected graph of $D_{M,r}$ obtained by replacing each directed edge with the corresponding undirected edge. 
 We show that $\widehat{G}_{M,r}$ is isomorphic to the graph $\overline{H_{M,r}}$.
  Note first that there is a natural bijection \[\varphi:
 V(\widehat{G}_{M,r})\to V(\overline{H_{M,r}})\] 
 between their vertex sets:
 a vertex $v$ in $V(\widehat{G}_{M,r})$ is a column of $M$ whose entry $r$ equals $1$ and therefore corresponds to a unique vertex $\varphi(v)$ of $V(\overline{H_{M,r}}) = V(H_{M,r}) = V(D_{M,r})$, namely to its support set.
 By definition, two vertices of $\widehat{G}_{M,r}$ are adjacent if and only if the corresponding columns are in conflict, that is, if there exist three rows of the matrix such that when restricted to the two columns, their entries are $(1,1)$, $(1,0)$, and $(0,1)$, respectively.
 Since the existence of a row with entries $(1,1)$ is guaranteed by row $r$, two columns are in conflict if and only if their support sets are incomparable with respect to inclusion.
 Thus, two vertices $u$ and $v$ are adjacent in $\widehat{G}_{M,r}$ if and only if the corresponding vertices $\varphi(u)$ and $\varphi(v)$ of $H_{M,r}$ are not adjacent.
 This shows that $\widehat{G}_{M,r}$ is isomorphic to $\overline{H_{M,r}}$, as claimed.

 Since $\widehat{G}_{M,r}$ is a perfect graph (see~\cite{CompAndAlg}), its chromatic number equals the maximum size of a clique.
 As shown above, the graphs $\widehat{G}_{M,r}$ and $\overline{H_{M,r}}$ are isomorphic to each other, hence, the maximum sizes of their cliques coincide.
 Moreover, the maximum size of a clique in $\overline{H_{M,r}}$ equals the maximum size of an independent set in $H_{M,r}$, which equals the maximum size of an antichain in $D_{M,r}$, since an independent set in $H_{M,r}$ is an antichain in $D_{M,r}$ and vice versa.
 Hence $\chi(G_{M,r}) =  \chi(\widehat{G}_{M,r}) = \wdt(D_{M,r})$, as claimed.
 \end{proof}

For a binary matrix $M$, we denote $\WW(M) = \sum_{r\in R_M}\wdt(D_{M,r})$. 
\Cref{thm:HRLB} and \Cref{lemma:principalsubgr} imply the following.

\smallskip
\begin{corollary}\label{cor:LB}
	Every binary matrix $M$ satisfies $\beta(M) \ge \WW(M)$.
\end{corollary}
\smallskip

We now introduce another polynomially computable lower bound on $\WW(M)$ (and consequently on $\beta(M)$). 
Let us define a vertex weight function $w\colon V(D_M)\to \mathbb{Z}_+$ by setting the weight of a vertex $v \in V(D_M)$ simply as the cardinality of the vertex, i.e., $w(v) = \vert v \vert $. 
Recall that we denote by $\alpha_w(M)$ the maximum weight, with respect to $w$, of an antichain in $D_M$.

\smallskip
\begin{proposition}\label{prop:newLB}
Every binary matrix $M$ satisfies  $\alpha_w(M)\le \WW(M)$.
\end{proposition}

\begin{proof}
Let $M$ be an arbitrary binary matrix and let $D_M$ be its containment digraph.
Fix an arbitrary antichain $X$ of maximum weight with respect to $w$ in $D_M$.
Let $R$ be the set of rows $r\in R_M$ such that $r\in v$ for some $v\in X$.
For a row $r\in R$, let $X_r$ be the set of vertices $v\in X$ containing $r$.
Note that $X_r$ is an antichain in $D_{M,r}$, the principal subgraph of $D_M$ corresponding to $r$, and hence $|X_r| \le \wdt(D_{M,r})$.
It follows that 
\begin{equation*}
\begin{split}
    \alpha_w(M) &= \sum_{v\in X} |v| = \sum_{r\in R}|X_r|\le \sum_{r\in R}\wdt(D_{M,r}) \le \sum_{r\in R_M}\wdt(D_{M,r}) = \WW(M)\,
    \end{split}
\end{equation*}
as claimed.
\end{proof}

Notice that the lower bound introduced above is indeed polynomial-time computable by \Cref{thm:max-weight-antichain}.

In summary, combining \Cref{prop:newLB,cor:LB,obs:MULB-MUB} yields the following chain of inequalities for any binary matrix $M$.
\smallskip

\begin{corollary}
	\label{cor:4ineq}
Every binary matrix $M$ satisfies
\begin{equation}\label{inequalities}
\alpha_w(M)\le \WW(M) \le \beta(M) \le \beta_{\ell}(M).    
\end{equation}
\end{corollary}

\section{Polynomially solvable cases}
\label{sec:polsolcases}

In this section, we identify new polynomially solvable cases of the MUB problem.
In \Cref{sec:WidthEqMaxEl}, we apply the results from \Cref{sec:bound} to show that the MUB problem is solvable in polynomial time on instances where the number of maximal elements of the corresponding poset $P_M$ equals the width of $M$.
Then, in \Cref{sec:boundedwidth}, we show that for any fixed integer $k\ge 1$, the MUB problem is solvable in polynomial time on instances of width $k$.

\subsection{Instances with many maximal elements}
\label{sec:WidthEqMaxEl}

Given a binary matrix $M$, we say that two vertices $v_1,v_2$ of the containment digraph $D_M$ are \emph{comparable} if either $v_1 \subset v_2$ or $v_2 \subset v_1$, otherwise we say that $v_1$ and $v_2$ are \emph{incomparable}.
Since no two sinks of $D_M$ are comparable, the set of maximal elements of the poset $P_M$ forms an antichain; in particular, the number of maximal elements of $P_M$ is at most the width of $M$.
Next, we show that in instances where this upper bound is achieved, all the four quantities appearing in the chain of inequalities~\eqref{inequalities} are the same.
In particular, this implies that for such instances the MUB problem can be solved in polynomial time.

Let $D_M$ be a containment digraph and let $B$ be a branching in $M$.
For every $v \in V(D_M)$, we denote by $\nu_B(v)$ the number of rows $r\in v$ that are uncovered in $v$ with respect to $B$.
Note that then $|U(B)| = \sum_{v \in V(D_M)} \nu_B(v)$. 
Let $H$ be a subgraph of $D_M$. 
A vertex $u\in V(H)$ is a \emph{maximal element} of $H$ if $d^+_H(u) = 0 $. 
Similarly, a vertex $u\in V(H)$ is a \emph{minimal element} of $H$ if $d^-_H(u) = 0$. 
We denote the set of maximal (respectively~minimal) elements by $\max \, H$ (respectively~$\min \, H$).
We extend the definition in the natural way also to sets of arcs of spanning subgraphs of $D_M$; in particular, to branchings.

\smallskip
\begin{lemma}\label{lem:linear-branchings}
Let $D_M$ be a containment digraph and let $B$ be a linear branching in $M$.
Then $|U(B)| = \sum_{v\in \max B} |v|$. 
\end{lemma}

\begin{proof}
By \Cref{obs:chain-partitions-and-linear-branchings}, there is a unique chain partition $P =\{C_1,\ldots, C_q\}$ of $D_M$ such that 
$B^P = B$.
For each $i\in \{1,\ldots, q\}$, let us denote by $k_i$ the cardinality of $C_i$.
Let us order the vertices of $C_i$ as $\{v_1^i,\ldots,v_{k_i}^i\}$ so that 
\hbox{$v^i_1\subset v^i_2 \subset \cdots \subset v_{k_i}^i$} (see \cite[Figure 21]{Baghirova} with $n=q$ in our case).
We count the number of uncovered pairs in each chain separately. 
For each $i\in\{1,\ldots,q\}$ and $j\in\{1,\ldots,k_i\}$, the number of uncovered pairs with second coordinate equal to $v^i_j$ is $\nu_B(v^i_j) = \vert v_j^i\setminus v_{j-1}^i \vert = |v^i_j| - |v^i_{j-1}|$ (where $v_0^i=\emptyset$).
Hence, the total number of uncovered pairs in $C_i$ equals \[\sum_{j=1}^{k_i}\nu_B(v^i_j) =\sum_{j=1}^{k_i} (|v^i_j| - |v^i_{j-1}|) = |v_{k_i}^i|\,.\]
The same argument holds for every chain in $P$, and thus the total number of uncovered pairs with respect to $B$ is
\[|U(B)| = \sum_{v\in V(D_M)}\nu_B(v)= \sum_{i = 1}^q |v_{k_i}^i| = \sum_{v\in \max B} |v|\,,\]
as claimed.
\end{proof}

Before we move on to the theorem, we state the results used for showing that the algorithm introduced in the proof of \Cref{thrm:wdtmaxel} runs in polynomial time. 
By Dilworth's theorem~\cite{dilworth}, the width of a DAG equals the minimum size of a chain partition in the digraph. 
Such a chain partition can be computed by reducing the problem to the maximum flow problem (see \cite{dilworthtomaxflow});\footnote{In fact, applying the approach of Fulkerson~\cite{Fulkerson} (see also Mäkinen et al.~\cite{makinen2015genome} and Ntafos and Hakimi~\cite{MR0545530}) reduces the problem on an $n$-vertex DAG to a maximum matching problem in a derived bipartite graph with $2n$ vertices.}
hence, the aforementioned  $\mathcal{O}(m^{1+o(1)})$ algorithm for the maximum flow problem due to van den Brand et al.~\cite{maxflowlinear} applies.
\smallskip

\begin{theorem}\label{thrm:wdtmaxel}
Let $M$ be a binary matrix with $m$ rows and $n$ columns such that the corresponding poset $P_M$ has exactly $q$ maximal elements, where $q =\wdt(M)$. 
 Then, $\alpha_w(M)= \WW(M)= \beta(M)=\beta_{\ell}(M)$, and an optimal branching of $D_M$ can be computed in time $\mathcal{O}(n^2(m+n^{o(1)}))$.
\end{theorem}

\begin{proof}
In order to show that $\alpha_w(M)= \WW(M)= \beta(M)=\beta_{\ell}(M)$, by \Cref{cor:4ineq}, it suffices to show that $\beta_{\ell}(M)\le \alpha_w(M)$.
To this end, we introduce a branching of the containment digraph $D_M$ corresponding to the chain partition of $D_M$ into $q$ chains and using such branching we show that $\beta_{\ell}(M)\le \alpha_w(M)$.
Furthermore, we show how to obtain such a branching, which will be optimal, in time $\mathcal{O}(n^2(m+n^{o(1)}))$.

Let $D_M$ be the containment digraph corresponding to $M$. Since $\wdt(M)=q$, by Dilworth's theorem, there exists a chain partition of $D_M$ into $q$ chains and, as mentioned earlier, such a chain partition can be found by reducing the problem to the maximum flow problem and solved in time $\mathcal{O}(|A(D_M)|^{1+o(1)})$, which is in $\mathcal{O}(n^{2+o(1)})$.  
Let $P=\{C_1,\ldots, C_q\}$ be a chain partition of $D_M$ into $q$ chains. 
For each $i\in \{1,\ldots, q\}$, the vertices of $C_i$ can be ordered as $\{v_1^i,\ldots,v_{k_i}^i\}$ where $k_i$ is the number of vertices in $C_i$ so that \hbox{$v^i_1\subset v^i_2 \subset \cdots \subset v_{k_i}^i$}.
Since, by assumption, there are $q$ maximal elements in $P_M$, and no two maximal elements can belong to the same chain, we conclude that the maximal elements of $P_M$ are the vertices $v^1_{k_1},\ldots, v^q_{k_q}$. It is easy to see that they form an antichain of maximal weight. Denote the weights of vertices $v_1^i,\ldots,v_{k_i}^i$ in $C_i$ by $w_1^i,\ldots,w_{k_i}^i$, respectively. Then, $\alpha_w(M) = w_{k_1}^1+\ldots+ w_{k_q}^q$.

By \Cref{lem:linear-branchings}, the number of uncovered pairs in the linear branching $B^P$ of $D_M$ is equal to $w_{k_1}^1 + \ldots +  w_{k_q}^q$.
Therefore, 
\[\beta_{\ell}(M)\le \vert U(B^P) \vert = w_{k_1}^1 + \ldots +  w_{k_q}^q =  \alpha_w(M)\,,\] 
where the inequality holds by the definition of linear uncovering number.
Recall that the containment digraph of a binary matrix $M\in \{0,1\}^{m \times n}$ can be computed in time $\mathcal{O}(n^2m)$. 
The linear branching corresponding to the chain partition of $D_M$ can be computed in time $\mathcal{O}(n^{2+o(1)})$.
Hence, the total time complexity is of the order $\mathcal{O}(n^2(m+n^{o(1)}))$.
\end{proof}

Notice that the complexity of recognizing whether $M$ is a binary matrix such that the corresponding poset $P_M$ has exactly $\wdt(M)$ maximal elements, is $\mathcal{O}(n^2m)$. 
First, as discussed above, the containment digraph of a binary matrix $M\in \{0,1\}^{m \times n}$ can be computed in time $\mathcal{O}(n^2m)$.
Computing its width can be done in time $\mathcal{O}(n^{2+o(1)})$.
Finally, the complexity of finding the maximal elements of $D_M$ is $\mathcal{O}(n)$, by considering each vertex and checking whether it has an outgoing arc. 
Hence the total time complexity is of the order $\mathcal{O}(n^2(m+n^{o(1)}))$.

\subsection{Instances of bounded width}
\label{sec:boundedwidth}
In this subsection, we show our main result: for any integer $k\ge 1$, there exists a polynomial-time algorithm for computing $\beta(M)$ for binary matrices $M$ of width at most $k$ (\Cref{thm:boundedwidth}). 
The exponent of the polynomial depends on~$k$.
Let us remark that recognizing if a given binary matrix $M$ has width at most $k$ can be done in \textsf{FPT} time by computing the containment digraph and applying to it the algorithm by C{\'a}ceres, Cairo, Mumey, Rizzi, and Tomescu~\cite{dilworthtomaxflow}.

Before we move on to the technical results of the section, let us first give an informal description of the algorithm. 
Given a binary matrix $M$ of width at most $k$, the algorithm seeks to construct an optimal branching of the corresponding digraph $D_M$ with at most $k$ leaves (that is, vertices of in-degree~0), whose existence follows from \Cref{lem:crossingarcs,lem:inter}. 
Such a branching necessarily contains at most $k$ vertices of in-degree at least two and any such vertex in fact has in-degree at most $k$. 
Based on this observation, the algorithm first tries to guess which vertices will have in-degree at least two in the sought branching, as well as their in-neighborhoods (such a vertex together with its in-neighborhood is referred to as an \emph{in-star}, see p.~\pageref{page-in-star}). 
For each guess, it then tries to complete the corresponding branching into a maximal branching with at most $k$ leaves and where all vertices of in-degree at least two are precisely those that have been guessed (this is referred to as a \emph{linear completion}, see p.~\pageref{page-linear-completion}). 
As we show in \Cref{lem:polylincomp}, deciding whether a guess may be completed into such a branching amounts to finding a perfect matching in an auxiliary bipartite graph and can be done in time polynomial in $|V(D_M)|$. Thus, the computationally expensive step is the first guessing phase which simply considers all possibilities, hence the $|V(D_M)|^{{\mathcal O}(k^2)}$ running time. 

Let us now define some relevant notions. 
Let $M$ be a binary matrix and $D_M$ the corresponding digraph. 
Then, for any maximal antichain $N$, we denote by $V^-_N \subseteq V(D_M)$ the set of vertices outside $N$ that can reach a vertex in $N$ and by $V_N^+\subseteq V(D_M)$ the set of vertices outside $N$ that are reachable from $N$. 
Note that by the maximality of $N$, the sets $N$, $V^-_N$, and $V_N^+$ form a partition of the vertex set of $D_M$.
 

Let $B$ be a branching of $D_M$. 
For any set $X\subseteq V(D_M)$, the \emph{restriction of $B$ to $X$} is the set $\{(u,v)\in B \colon u,v\in X\}$.
We say that $B$ is a \emph{maximal} branching if every vertex in $V(D_M) \setminus (\max \, D_M)$ has out-degree one in $B$. 
Given a maximal antichain $N$ of $D_M$, an arc $(u,v)\in B$ is a \emph{crossing arc with respect to $N$} if $u\in V^-_N$ and $v \in V_N^+$, and a \emph{non-crossing arc with respect to $N$} otherwise. 
If the antichain $N$ is clear from the context, we may simply say that $(u,v)$ is a \emph{crossing arc} and a \emph{non-crossing arc}, respectively.


The aim of the following three lemmas is to show that given a binary matrix of width $k$, there exists an optimal branching of the corresponding digraph $D_M$ with at most $k$ leaves. The first step consists in showing that $D_M$ has an optimal branching $B$ such that its restriction to $V_N^- \cup N$, for a given maximum antichain $N$ of $D_M$, is a linear branching.

\smallskip
\begin{lemma}
\label{lem:crossingarcs}
For any containment digraph $D_M$ and any maximum antichain $N$ of $D_M$, there exists an optimal branching $B$ of $D_M$ such that $B$ contains no crossing arc with respect to~$N$ and the restriction of $B$ to $V^-_N \cup N$ is a linear branching.
\end{lemma}

\begin{proof}
Let $D_M$ be a containment digraph and let $N$ be a maximum antichain of $D_M$. Consider an optimal maximal branching $B$ of $D_M$ and denote by $B_C = \{(u,v) \in B\colon u \in V_N^- \text{ and }v \in V_N^+\}$ the set of crossing arcs with respect to~$N$ (note that $B_C$ may be empty). 
Let $P$ be a chain partition of the subgraph of $D_M$ induced by $V_N^- \cup N$ into $|N|$ chains and let $B^P$ be the corresponding linear branching of $V_N^- \cup N$. 
Denoting by $B^\top$ the restriction of $B$ to $N \cup V_N^+$, let us show that the branching $B' = B^\top \cup B^P$ is optimal. Since $B'$ contains no crossing arc with respect to~$N$ and the restriction of $B'$ to $V^-_N \cup N$ is a linear branching, this will conclude the proof.

Denote by $\nu_B^- = \sum_{u \in V_N^- \cup N} \nu_B(u)$ and by $\nu_B^+ = \sum_{u \in V_N^+} \nu_B(u)$.
Note that $|U(B)| = \nu_B^- + \nu_B^+$. 
Let us show that $|U(B')| \leq \nu_B^- + \nu_B^+$. 
Since $\lvert  \bigcup_{(v,u) \in B} v\rvert  \leq \sum_{(v,u) \in B} \lvert v\rvert $ for every $u \in V_N^- \cup N$, we have that 
\begin{equation*}
\begin{split}
\nu_B^- &= \sum_{u \in V_N^- \cup N} \Bigg(  \lvert u \rvert - \Bigg\lvert \bigcup_{(v,u) \in B} v \Bigg\rvert \Bigg) \geq \sum_{u \in V_N^- \cup N} \Bigg(\lvert u\rvert  - \sum_{(v,u) \in B} \lvert v\rvert \Bigg) \\
&= \sum_{u \in V_N^- \cup N} \lvert u\rvert  - \sum_{u \in V_N^- \cup N}\sum_{(v,u) \in B} \lvert v\rvert\,.
\end{split}
\end{equation*}

Then denoting by $V_C^- = \{u \in V_N^-\colon  (u,v) \in B_C \text{ for some } v \in N^+_B(u)\}$ the set of tails of arcs in $B_C$,
\[
\sum_{u \in V_N^- \cup N}\sum_{(v,u) \in B} \lvert v \rvert = \sum_{v \in V_N^- \setminus V_C^-} \lvert v \rvert\,,
\]
since $V_N^- \setminus V_C^- = \{v\in V_N^- \colon (v,u) \in B \text{ for some } u \in V_N^- \cup N\}$ (note indeed that by maximality of $B$, any vertex in $V_N^-$ is either the tail of a crossing arc or the tail of a non-crossing arc). 
It follows that
\begin{equation*}
\begin{split}
 \nu^-_B &\geq \sum_{u \in V_N^- \cup N} \lvert u\rvert  - \sum_{u \in V_N^- \setminus V_C^-} \lvert u\rvert  = \sum_{u \in N \cup V_C^-} \lvert u\rvert \,.
\end{split}
\end{equation*}
Next let $D_C^+ = \{v \in V_N^+\colon \exists u \in N^-_B(v)$ such that $(u,v) \in B_C\}$ be the set of heads of arcs in $B_C$. 
Then
\begin{equation*}
	\begin{split}
		\nu_B^+ &= \sum_{u \in V_N^+ \setminus D_C^+} \nu_B(u) + \sum_{u \in D_C^+} \nu_B(u) = \sum_{u \in V_N^+ \setminus D_C^+} \nu_B(u) + \sum_{u \in D_C^+} \Bigg(|u| - \Bigg\lvert\bigcup_{(v,u) \in B} v \Bigg\rvert\Bigg)\,,\\
	\end{split}
\end{equation*}
where for every $u \in D_C^+$,
\begin{equation*}
|u| - \Bigg\lvert\bigcup_{(v,u) \in B} v \Bigg\rvert
		\geq |u| - \Bigg(\Bigg\lvert\bigcup_{(v,u) \in B \setminus B_C} v\Bigg\rvert + \Bigg\lvert\bigcup_{(v,u) \in B_C} v \Bigg\rvert \Bigg) \,,
\end{equation*}
since $|\bigcup_{(v,u) \in B} v| \leq |\bigcup_{(v,u) \in B \setminus B_C} v| + |\bigcup_{(v,u) \in B_C} v|$. 
Since $\nu_B(u) = \nu_{B^\top}(u)$ for every $u \in V_N^+ \setminus D_C^+$, and for every $u \in D_C^+$, it holds that
\[
\lvert u\rvert -  \Bigg\lvert \bigcup_{(v,u) \in B \setminus B_C} v \Bigg\rvert = \lvert u \rvert -  \Bigg\lvert\bigcup_{(v,u) \in B^\top} v \Bigg\rvert = \nu_{B^\top}(u),
\]
it follows that
\begin{equation*}
\begin{split}
\nu_B^+ &\geq \sum_{u \in V_N^+ \setminus D_C^+} \nu_{B^\top}(u) + \sum_{u \in D_C^+} \nu_{B^\top}(u) - \sum_{u \in D_C^+} \Bigg\lvert \bigcup_{(v,u) \in B_C} v \Bigg\rvert\\
&\geq \sum_{u \in V_N^+} \nu_{B^\top}(u) - \sum_{u \in D_C^+}\sum_{(v,u) \in B_C} \lvert v \rvert = \sum_{u \in V_N^+} \nu_{B^\top}(u)  - \sum_{v \in V_C^-} \lvert v \rvert\,,
\end{split}
\end{equation*}
since $|\bigcup_{(v,u) \in B_C} v| \leq \sum_{(v,u) \in B_C} |v|$ for every $u \in D_C^+$. Now observe that, by \Cref{lem:linear-branchings},
\[
\sum_{v \in V_N^- \cup N} \nu_{B'}(v) = \sum_{v \in V_N^- \cup N} \nu_{B^P}(v) = \sum_{v \in N} \lvert v \rvert \,,
\]
since $B^P$ is a linear branching of the subgraph induced by $V_N^- \cup N$, and so
\begin{equation*}
\begin{split}
|U(B')| &= \sum_{v \in V_N^- \cup N} \nu_{B'}(v) + \sum_{v \in V_N^+} \nu_{B'}(v)= \sum_{v \in N} \lvert v \rvert + \sum_{v \in V_N^+} \nu_{B^\top}(v) \\
&= \sum_{v \in N} \lvert v\rvert + \sum_{v \in V_C^-} \lvert v \rvert + \sum_{v \in V_N^+} \nu_{B^\top}(v) - \sum_{v \in V_C^-} \lvert v \rvert\leq \nu_B^- + \nu_B^+ = |U(B)|.
\end{split}
\end{equation*}
Thus, $B'$ is an optimal branching of $D_M$, which concludes the proof.
\end{proof}

The following lemma shows that, given a maximum antichain $N$ of $D_M$, the set $V_N^-$ is irrelevant for the computation of $\beta(D_M)$.

\smallskip
\begin{lemma}
\label{lem:reduce}
Let $D_M$ be a containment digraph and let $N$ be a maximum antichain of $D_M$. Then $$\beta(D_M) = \beta(D_M[N \cup V_N^+]).$$
\end{lemma}

\begin{proof}
By \Cref{lem:crossingarcs}, there exists an optimal branching $B$ of $D_M$ such that $B$ contains no crossing arcs with respect to~$N$ and the restriction of $B$ to $V^-_N \cup N$ is a linear branching. 
Let us denote by $B^L$ and $B^\top$ the restrictions of $B$ to $V^-_N \cup N$ and $N \cup V_N^+$, respectively. 
Since $B^L$ is a linear branching with $\max B^L = N$, \Cref{lem:linear-branchings} implies that \[\sum_{v\in V_N^-\cup N}\nu_{B^L}(v) = |U(B^L)| = \sum_{v\in N}|v|\] and consequently $\sum_{v\in V_N^-\cup N}\nu_B(v) = \sum_{v\in N}|v|$.
It follows that
\begin{equation*}
\begin{split}
|U(B)| &= \sum_{v \in V^-_N \cup N} \nu_B(v) + \sum_{v \in V_N^+} \nu_B(v)= \sum_{v \in N} \lvert v \rvert + \sum_{v \in V_N^+} \nu_B(v)= \sum_{v \in N} \nu_{B^\top}(v) + \sum_{v \in V_N^+} \nu_{B^\top}(v)\,.
\end{split}
\end{equation*}
We claim that $B^\top$ is an optimal branching of $D^\top = D_M[N \cup V_N^+]$. Indeed, suppose to the contrary that $B^\top$ is not an optimal branching of $D^\top$ and let $B^*$ be an optimal branching of $D^\top$. 
Then
\[
\sum_{v \in N \cup V_N^+} \nu_{B^\top}(v) > \sum_{v \in N \cup V_N^+} \nu_{B^*}(v)\,.
\]
Now $B' = B^L \cup B^*$ is a branching of $D_M$ and thus, similarly to the above,
\begin{equation*}
\begin{split}
\lvert U(B') \rvert &= \sum_{v \in V^-_N \cup N} \nu_{B'}(v) + \sum_{v \in V_N^+} \nu_{B'}(v)= \sum_{v \in N} \lvert v \rvert + \sum_{v \in V_N^+} \nu_{B'}(v)\\
&= \sum_{v \in N} \nu_{B^*}(v) + \sum_{v \in V_N^+} \nu_{B^*}(v)< \lvert U(B) \rvert\,,
\end{split}
\end{equation*}
a contradiction to the optimality of $B$. Therefore, $B^\top$ is an optimal branching of $D^\top$ and thus,
\[\beta(D^\top) = \sum_{v \in N \cup V_N^+} \nu_B(v) = \lvert U(B) \rvert = \beta(D_M)\] as claimed.
\end{proof}

The following lemma is the last and key step in showing that $D_M$ has a branching with $\wdt(D_M)$ leaves.

\smallskip
\begin{lemma}
\label{lem:inter}
Let $D_M$ be a containment digraph such that $\min D_M$ is the only maximum antichain of $D_M$. 
Then there exists a maximal optimal branching $B$ whose leaves are exactly the minimal elements of $D_M$.
\end{lemma}

\begin{proof}
It suffices to show that there exists a maximal optimal branching of $D_M$ with at most $\wdt(D_M)$ leaves.
Let $B$ be a maximal optimal branching of $D_M$ and suppose that $B$ has strictly more than $\wdt(D_M)$ leaves. 
Let us show how to modify $B$ to obtain a maximal optimal branching of $D_M$ with fewer leaves. To this end, we first prove the following.

\begin{claim}
\label{clm:inneighbors}
For every vertex $v \in V(D_M) \setminus \min D_M$, we have $|N^-_{D_M}(v)| \geq 2$. 
\end{claim}

\begin{sloppypar}
\begin{claimproof}
Consider a vertex $v \in V(D_M) \setminus \min D_M$ and suppose for contradiction that \hbox{$|N^-_{D_M}(v)| \le 1$.}
Then $|N^-_{D_M}(v)| = 1$, since $v$ is not a minimal element of $D_M$, hence, $v$ is comparable to exactly one element $u \in \min D_M$, which implies that $((\min D_M) \setminus \{u\}) \cup \{v\}$ is a maximum antichain of $D_M$, a contradiction to the fact that $\min D_M$ is the unique maximum antichain of~$D_M$.
\end{claimproof}
\end{sloppypar}

An \emph{alternating sequence} (with respect to $B$) is a sequence of distinct arcs $(u_1,v_1),\ldots,(u_p,v_p) \in A(D_M)$ satisfying the following conditions.
\begin{itemize}
\item For every $\ell \in [p]$ such that $\ell$ is odd, $(u_\ell, v_\ell) \notin B$.
\item For every $\ell \in [p]$ such that $\ell$ is even, $(u_\ell,v_\ell) \in B$.
\item For every $\ell \in [p-1]$ such that $\ell$ is odd, $u_\ell = u_{\ell +1}$.
\item For every $\ell \in [p-1]$ such that $\ell$ is even, $v_\ell = v_{\ell+1}$.
\end{itemize}
The alternating sequence $(u_1,v_1),\ldots,(u_p,v_p)$ is said to \emph{start} in $v_1$; it is of \emph{odd length} if $p$ is odd and of \emph{even length} otherwise. Let us note that, though the arcs of the sequence are assumed to be pairwise distinct, we do not forbid two non-consecutive arcs in the sequence from having a common endvertex.
See \Cref{fig:altseq} for an illustration of an alternating sequence.

 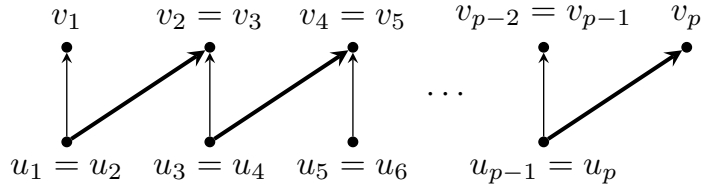
\begin{figure}[htb]
 \centering
 \resizebox{0.6\textwidth}{!}{%
 \begin{tikzpicture}
 \node[mcirc,label=below:{\small $u_1 = u_2$}] (u1) at (0,0) {};
 \node[mcirc,label=above:{\small $v_1$}] (v1) at (0,1) {};
 \node[mcirc,label=below:{\small $u_3 =u_4$}] (u2) at (1.5,0) {};
 \node[mcirc,label=above:{\small $v_2=v_3$}] (v2) at (1.5,1) {};
 \node[mcirc,label=below:{\small $u_5 = u_6$}] (u3) at (3,0) {};
 \node[mcirc,label=above:{\small $v_4 = v_5$}] (v3) at (3,1) {};
 \node[draw=none] at (4,.5) {\small $\cdots$};
 \node[mcirc,label=below:{\small $u_{p-1} = u_p$}] (up2) at (5,0) {};
 \node[mcirc,label=above:{\small $v_{p-2} = v_{p-1}$}] (vp2) at (5,1) {};
 \node[mcirc,label=above:{\small $v_p$}] (vp1) at (6.5,1) {};

 \draw[->,>=stealth] (u1) -- (v1);
 \draw[->,>=stealth] (u2) -- (v2);
 \draw[->,>=stealth] (u3) -- (v3);
 \draw[->,>=stealth] (up2) -- (vp2);
 \draw[very thick,->,>=stealth] (u1) -- (v2);
 \draw[very thick,->,>=stealth] (u2) -- (v3);
 \draw[very thick,->,>=stealth] (up2) -- (vp1);
 \end{tikzpicture}
 }
 \caption{An alternating sequence $(u_1,v_1),\ldots,(u_p,v_p)$ of even length (bold arcs belong to~$B$).}
 \label{fig:altseq}
 \end{figure}

An alternating sequence $(u_1,v_1),\ldots,(u_p,v_p)$ is said to be \emph{good} if $p$ is even and the following two conditions hold.
\begin{itemize}
\item[(i)] $d^-_B(v_1) = 0$ and $d^-_B(v_p) \geq 2$,
\item[(ii)] $d^-_B(v_\ell) = 1$ whenever $2 \leq \ell \leq p-1$.
\end{itemize}

\begin{claim}
\label{clm:altseq}
If there exists a good alternating sequence $(u_1,v_1),\ldots,(u_p,v_p)$, then 
\[(B \setminus \{(u_i,v_i)\colon i \in [p] \text{ is even}\}) \cup \{(u_i,v_i)\colon i \in [p] \text{ is odd}\}\]
is a maximal optimal branching of $D_M$ with fewer leaves than $B$.
\end{claim}

\begin{claimproof}
Assume that $\sigma=(u_1,v_1),\ldots,(u_p,v_p)$ is a good alternating sequence and denote by $B^\sigma = (B \setminus \{(u_i,v_i)\colon i \in [p] \text{ is even}\}) \cup \{(u_i,v_i)\colon i \in [p] \text{ is odd}\}$. 
Clearly, $B^\sigma$ is a maximal branching of $D_M$.
Furthermore, since $v_1$ is a leaf in $B$ but not in $B^\sigma$, and every leaf in $B^\sigma$ is a leaf in $B$, the branching $B^\sigma$ has fewer leaves than $B$. 
It remains to show that $B^\sigma$ is optimal. 
To see this, observe that by condition (ii), any two vertices $v_i$ and $v_j$ with even indices $i\neq j$ are distinct, and therefore
\begin{equation*}
\begin{split}
\lvert U(B)\lvert & = \sum_{v \in V(D_M) \setminus \{v_i\colon i \in [p]\}} \nu_B(v) + \sum_{i \in [p]\colon
\text{$i$ even}} \nu_B(v_i) + \lvert v_1 \rvert,
\end{split}
\end{equation*}
where
\begin{equation*}
\begin{split}
\sum_{i \in [p]\colon
\text{$i$ even}} \nu_B(v_i)
& = \sum_{i \in [p-2]\colon
\text{$i$ even}}\left(\lvert v_i \rvert - \lvert u_i\rvert\right) + \nu_B(v_p) \\
& = \sum_{i \in [p-2]\colon
\text{$i$ even}} \left(\lvert v_{i+1}\rvert - \lvert u_{i-1} \rvert \right)+ \nu_B(v_p)
\end{split}
\end{equation*}
as $v_i = v_{i+1}$ and $u_i = u_{i-1}$ for every $i \in [p-1]$ such that $i$ is even. Since
\begin{equation*}
\begin{split}
\nu_B(v_p) & = \lvert v_p \rvert - \Bigg\lvert \bigcup_{v \in N^-_B(v_p)} v\Bigg\rvert \geq \lvert v_p \rvert - \Bigg\lvert \bigcup_{v \in N^-_B(v_p) \setminus \{u_p\}} v\Bigg\rvert  - \lvert u_p \rvert\\
& = \lvert v_p\vert - \Bigg\lvert \bigcup_{v \in N^-_{B^\sigma}(v_p)} v\Bigg\rvert  - \lvert u_p \rvert = \nu_{B^\sigma}(v_p) - \lvert u_p \rvert
\end{split}
\end{equation*}
and $u_p = u_{p-1}$ by definition, it follows that
\begin{equation*}
\begin{split}
&\sum_{i \in [p-2]\colon
\text{$i$ even}} \left(\lvert v_{i+1}\rvert - \lvert u_{i-1} \rvert \right)+ \nu_B(v_p)
\geq  \sum_{i \in [p-2]\colon\text{$i$ even}} (\lvert  v_{i+1}\rvert   - \lvert  u_{i-1}\rvert)   + \nu_{B^\sigma}(v_p) - \lvert  u_{p-1}\rvert \\
= & \sum_{i \in [p-1]\colon \text{$i$ odd}} (\lvert  v_i\rvert   - \lvert  u_i\rvert)   + \nu_{B^\sigma}(v_p)  - \lvert v_1 \rvert
=  \sum_{i \in [p-1]\colon \text{$i$ odd}} \nu_{B^\sigma}(v_i) + \nu_{B^\sigma}(v_p)  - \lvert v_1 \rvert\,.
\end{split}
\end{equation*}
 Therefore, $\lvert U(B) \rvert$ is at least
\begin{equation*}
\begin{split}
&\sum_{v \in V(D_M) \setminus \{v_i\colon i \in [p]\}} \nu_B(v) + \sum_{i \in [p-1]\colon \text{$i$ odd}} \nu_{B^\sigma}(v_i) + \nu_{B^\sigma}(v_p)=  \lvert  U(B^\sigma)\rvert\,.
\end{split}
\end{equation*}
But $B$ is optimal and thus, $B^\sigma$ is optimal as well.
\end{claimproof}

Consider now a leaf $v \in V(D_M) \setminus \min D_M$ of $B$ (recall that by assumption, such a leaf exists) and denote by $\mathcal{S}_v$ the set of alternating sequences starting in $v$. 
Note that since $d^-_B(v) = 0$ and $v \notin \min D_M$, we infer that $\mathcal{S}_v \neq \emptyset$ (a single arc $e \in A(D_M) \setminus B$ with head $v$ indeed constitutes an alternating sequence). 
We will show that there exists a good alternating sequence in $\mathcal{S}_v$.
This will complete the proof, as by \Cref{clm:altseq}, we can then construct an optimal branching of $D_M$ with fewer leaves than $B$. 

Suppose for a contradiction that no good alternating sequence exists in $\mathcal{S}_v$. 
In the following, an alternating sequence $\sigma \in \mathcal{S}_v$ is said to be \emph{maximal} if $\sigma$ is not a proper subsequence of any alternating sequence in $\mathcal{S}_v$.
Let $\sigma=(u_1,v_1),\ldots,(u_p,v_p)$ be a maximal alternating sequence in $\mathcal{S}_v$ (note that $v_1 = v$). 

We claim that there exists no index $i \in [p]$ such that $d^-_B(v_i) \geq 2$. 
Indeed, suppose for a contradiction that such indices exist and let $i \in [p]$ be the smallest such index.
Then $i$ is even, for if $i$ is odd, then $i\ge 3$ and $v_i = v_{i-1}$ by definition, a contradiction to the minimality of $i$.
Furthermore, by minimality of $i$, $d^-_B(v_\ell) = 1$ for every $2 \leq \ell \leq i-1$. 
It follows that $(u_1,v_1),\ldots,(u_i,v_i)$ is a good alternating sequence in $\mathcal{S}_v$, a contradiction to our assumption. 
Thus, $d^-_B(v_i) \leq 1$ for every $i \in [p]$. 

We now claim that $p$ is odd. Indeed, suppose to the contrary that $p$ is even (note that in particular $p \geq 2$). 
Then $(u_p,v_p) \in B$ by definition and since $d^-_B(v_p) \leq 1$ by the above, it follows from \Cref{clm:inneighbors} that there exists a vertex $u \in N^-_{D_M}(v_p) \setminus \{u_p\}$. 
We claim that the arc $(u,v_p)$ is not contained in $\sigma$. 
Indeed, if there exists $\ell \in [p]$ such that $(u,v_p) = (u_\ell, v_\ell)$, then $\ell$ is odd, since $(u,v_p) \notin B$, and strictly larger than one, since $d^-_B(v_1) = 0$; but then, $(u_{\ell -1},v_{\ell -1}) = (u_p,v_p)$ as $d^-_B(v_p) \leq 1$, a contradiction to the fact that the arcs of $\sigma$ are pairwise distinct. 
Thus, $(u,v_p) \notin \sigma$ and so, $(u_1,v_1),\ldots,(u_p,v_p),(u,v_p)$ is an alternating sequence in $\mathcal{S}_v$, contradicting the maximality of $\sigma$. 
Therefore, $p$ is odd.

Since $p$ is odd, $(u_p,v_p) \notin B$. Now since $u_p \notin \max D_M$ and $B$ is maximal, there must exist a vertex $u \in N^+_{D_M}(u_p)$ such that $(u_p,u) \in B$; but $\sigma$ is maximal and so, the arc $(u_p,u)$ must be contained in $\sigma$. 
Thus, we have proven the following.

\begin{claim}
\label{obs:maxseq}
For every maximal alternating sequence $\sigma=(u_1,v_1),\ldots,(u_p,v_p)$ in $\mathcal{S}_v$, the following hold.
\begin{itemize}
\item[(i)]\label{item1} For every $i \in [p]$, $d^-_B(v_i) \leq 1$.
\item[(ii)]\label{item2} For every $i \in [p]$, $\sigma$ contains the arc $(u_i,u)$ where $u$ is the out-neighbor of $u_i$ in $B$.
\end{itemize}
\end{claim}

Now let $B_v \subseteq B$ be the set of arcs $(u,w) \in B$ such that there exists an alternating sequence in $\mathcal{S}_v$ containing the arc $(u,w)$ (note that if there exists an alternating sequence in $\mathcal{S}_v$ containing $(u,w)$ then, a fortiori, there exists a maximal alternating sequence in $\mathcal{S}_v$ containing $(u,w)$). 

We claim that $B_v$ is a linear branching. 
Indeed, suppose that there exists a vertex $w \in V(D_M)$ such that $d^-_{B_v}(w) \geq 2$ and let $u \in N^-_{B_v}(w)$ be an in-neighbor of $w$ in $B_v$ (that is, $(u,w) \in B_v$). 
Then, by construction, there exists a maximal alternating sequence in $\mathcal{S}_v$ containing the arc $(u,w)$; but $d^-_B(w) \geq d^-_{B_v}(w) \geq 2$, a contradiction to \Cref{obs:maxseq}(i). 

Now let $L_v \subseteq \min D_M$ be the set of vertices of out-degree one in $B_v$, that is, $u \in L_v$ if and only if $u$ is a minimal element of $D_M$ and there exists $w \in V(D_M)$ such that $(u,w) \in B_v$. 
Note that $L_v \neq \emptyset$ as it contains, in particular, every minimal element of $D_M$ comparable to $v$. 
For every vertex $u \in L_v$, let $m_u \in \max B_v$ be the maximal element reachable from $u$ in $B_v$, that is, there is a path from $u$ to $m_u$ using only arcs in $B_v$ and $m_u$ is the tail of no arcs in $B_v$. 
Note that $m_u \neq m_w$ for any two distinct vertices $u,w \in L_v$, as $B_v$ is a linear branching; in particular, $|\{m_u\colon u \in L_v\}| = |L_v|$. 
Denoting by $M_v = \{m_u\colon u \in L_v\}$, we will show that $(\min D_M \setminus L_v) \cup M_v$ is an antichain of $D_M$. 
Since this set has size $|\min D_M \setminus L_v| + |M_v| = |\min D_M \setminus L_v| + |L_v| = |\min D_M|$ as shown above, this contradicts the fact that $\min D_M$ is the unique maximum antichain of $D_M$ and thus concludes the proof.

Let us first show that the elements in $M_v$ are pairwise incomparable. 
Suppose to the contrary that there exist $m_u,m_w \in M_v$ such that $(m_u,m_w) \in A(D_M)$, and let $x \in V(D_M)$ be the in-neighbor of $m_w$ in $B_v$. 
Then $m_u \neq x$ for if $m_u = x$ then $m_u \notin \max B_v$, a contradiction to the definition of $m_u$. 
Furthermore, since the arc $(x,m_w)$ is contained in a maximal alternating sequence in $\mathcal{S}_v$ by construction, $(m_u,m_w) \notin B$ by \Cref{obs:maxseq}(i).
Now let $(u_1,v_1),\ldots,(u_p,v_p)$ be an alternating sequence in $\mathcal{S}_v$ ending with the arc $(x,m_w)$.
Note that $p$ is even as $(x,m_w) \in B$. 
Then $(u_1,v_1),\ldots,(u_p, v_p),(m_u,m_w)$ is an alternating sequence in $\mathcal{S}_v$; in particular, it is contained in some maximal alternating sequence $\sigma$ in $\mathcal{S}_v$. 
But then, by \Cref{obs:maxseq}(ii), $\sigma$ contains the arc $(m_u,y)$ where $y$ is the out-neighbor of $m_u$ in $B$, a contradiction to the fact that $m_u \in \max B_v$. Thus, the elements in $M_v$ are pairwise incomparable.

Second, let us show that every element in $\min D_M \setminus L_v$ is incomparable to every element in $M_v$. Suppose to the contrary that there exists $m \in \min D_M \setminus L_v$ and $m_u \in M_v$ such that $(m,m_u) \in A(D_M)$.
Note that $m_u\neq u$ by the definition of $L_v$ and hence $m_u$ has a (unique) in-neighbor in $B_v$. 
Let $x \in V(D_M)$ be the in-neighbor of $m_u$ in $B_v$ (note that since $m$ is the tail of no arc in $B_v$, necessarily $m \neq x$). 
By construction, there exists a maximal alternating sequence $(u_1,v_1),\ldots,(u_p,v_p)$ in $\mathcal{S}_v$ containing the arc $(x,m_u)$, say $(x,m_u) = (u_\ell, v_\ell)$ for some $\ell \in [p]$ (note that $\ell$ is even as $(x,m_u) \in B$). 
Since $(m,m_u) \notin B$ by \Cref{obs:maxseq}(i), it follows that $(u_1,v_1),\ldots,(u_\ell, v_\ell),(m,m_u)$ is an alternating sequence in $\mathcal{S}_v$; in particular, it is contained in some maximal alternating sequence $\sigma$ in $\mathcal{S}_v$. But then, by \Cref{obs:maxseq}(ii), $\sigma$ contains the arc $(m,y)$ where $y$ is the out-neighbor of $m$ in $B$, a contradiction to the fact that $m \notin L_v$, which concludes the proof.
\end{proof}

Given an integer $p\ge 2$, an \emph{in-star of size $p$}\label{page-in-star} in $D_M$ is a sequence $S = (c,x_1,\ldots,x_p)$ of $p+1$ pairwise distinct vertices of $D_M$ such that $(x_i,c) \in A(D_M)$ for every $i \in [p]$.
We say that two in-stars $(c_1,x^1_1,\ldots,x^1_{p_1})$ and $(c_2,x^2_1,\ldots,x^2_{p_2})$ are \emph{independent} if $c_1 \neq c_2$ and $\{x_1^1,\ldots,x_{p_1}^1\} \cap \{x_1^2,\ldots,x_{p_2}^2\} = \emptyset$. Note that it is possible that $c_1\in \{x_1^2,\ldots,x_{p_2}^2\}$ or $c_2\in \{x_1^1,\ldots,x_{p_1}^1\}$ (but not both).
See \Cref{fig:instars} for an example.

\begin{figure}[ht!]
     \centering
     \includegraphics[width=0.4\linewidth]{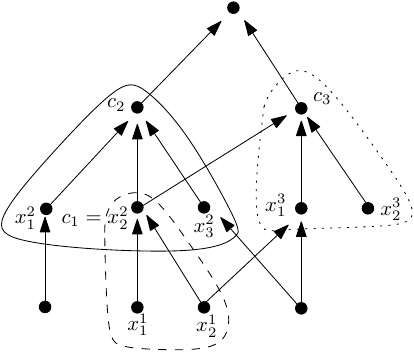}
     \caption{ An example of a containment digraph and 
     three pairwise independent in-stars with $p_1=2$, $p_2 = 3$, and $p_3 = 2$ depicted using dashed, normal, and dotted lines, respectively. 
     Note that the arcs that follow from transitivity are not shown; however, such arcs could also be a part of an in-star.}
     \label{fig:instars}
 \end{figure}

Let $\mathcal{S} = \{(c_i,x^i_1,\ldots,x^i_{p_i})\colon i \in [\ell]\}$ be a set of pairwise independent in-stars of $D_M$. 
Let $C = \{c_i\colon i \in [\ell]\}$ be the set of sinks of the in-stars, $X= \{x_j^i\colon i \in [\ell] \text{ and } j \in [p_i]\}$ be the set of sources of the in-stars, and $A(\mathcal{S}) = \{(x^i_j,c_i) \colon i \in [\ell] \text{ and } j \in [p_i]\}$ be the set of arcs of the in-stars.
A \emph{linear completion}\label{page-linear-completion} of $\mathcal{S}$ is a set $L \subseteq A(D_M)\setminus A(\mathcal{S})$ of arcs such that for every $u \in V(D_M)$, the following hold.
\begin{itemize}
\item If $u \notin C \cup X$ then $u$ is the tail of exactly one arc in $L$ (unless $u \in \max D_M$, in which case $u$ is the tail of no arc in $L$) and $u$ is the head of exactly one arc in $L$ (unless $u \in \min D_M$, in which case $u$ is the head of no arc in $L$).
\item If $u \in C \setminus X$ then $u$ is the tail of exactly one arc in $L$ (unless $u \in \max D_M$) and the head of no arc in $L$.
\item If $u \in X \setminus C$ then $u$ is the head of exactly one arc in $L$ (unless $u \in \min D_M$) and the tail of no arc in $L$.
\item If $u \in C \cap X$ then $u$ is the head of no arc in $L$ and the tail of no arc in $L$.
\end{itemize} 
 See \Cref{fig:lincompletion} for an example of a set of pairwise independent in-stars of $D_M$ that admits a linear completion and one that does not.

\begin{figure}[ht!]
    \centering
    \includegraphics[width=0.75\linewidth]{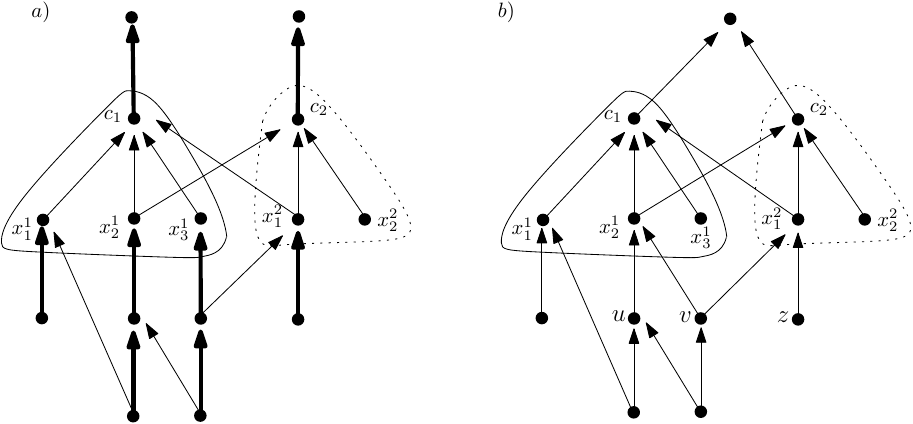}
     \caption{Two examples of a containment digraph (with the arcs that follow from transitivity not shown).
     In part a), two pairwise independent in-stars $(c_1,x_1^1,x_2^1,x_3^1)$ and $(c_2,x_1^2,x_2^2)$ are shown, admitting a linear completion represented by the bold arcs.
     On the other hand, the two pairwise independent in-stars depicted in part b) do not admit a linear completion, since the only outgoing arcs from $c_1$ and $c_2$ have the same head.
     It is also not possible to select one outgoing arc from each of $u$, $v$, and $z$ while respecting all the conditions from the definition of a linear completion.
     We emphasize that the arcs that follow from transitivity could also be part of a linear completion.}
  \label{fig:lincompletion}
 \end{figure}

In other words, $A(\mathcal{S}) \cup L$ is a maximal branching of $D_M$ where all vertices of in-degree at least two belong to~$C$.

In the following lemma, we consider the case when the only maximum antichain is the set of minimal elements (cf.~\Cref{lem:inter}) and show that in this case there exists a set of pairwise independent in-stars that admits a linear completion.
For a branching $B$ of $D_M$, we denote by $\mathcal{S}_B$ the set of in-stars formed by the vertices of in-degree at least two in $B$ and their in-neighborhoods (in $B$).
Note that any two in-stars in $\mathcal{S}_B$ are independent.

\smallskip
\begin{lemma}
\label{lem:forall}
Let $D_M$ be a containment digraph such that $\min D_M$ is the only maximum antichain of $D_M$.
Let $B$ be a maximal optimal branching $B$ of $D_M$ whose leaves are exactly the minimal elements of $D_M$.
Then, the following hold.
\begin{itemize}
\item[(i)] $\mathcal{S}_B$ admits a linear completion.
\item[(ii)] For any linear completion $L$ of $\mathcal{S}_B$, the set $A(\mathcal{S}_B) \cup L$ is an optimal branching of $D_M$.
\end{itemize}
\end{lemma}

\begin{proof}
Let $\mathcal{S}_B = \{(c_i,x^i_1,\ldots,x^i_{p_i})\colon  i \in [\ell]\}$. Further denote by $C = \{c_i\colon i\in [\ell]\}$ and by $X = \{x^i_j\colon  i \in [\ell] \text{ and } j \in [p_i]\}$. 
To prove (i), let us show that $B^L = B \setminus A(\mathcal{S}_B)$ is a linear completion of $\mathcal{S}_B$.
Consider a vertex $u \in V(D_M)$ and suppose first that $u \notin C \cup X$. 
If $u \notin \max D_M$ then $u$ is the tail of exactly one arc in $B^L$ by maximality of $B$; and if $u \notin \min D_M$ then $u$ is the head of exactly one arc in $B^L$ as the set of leaves of $B$ is exactly $\min D_M$. 
Using similar arguments, it is not difficult to verify that the properties imposed on $u$ as in the definition of linear completion also hold in the remaining three cases.

To prove (ii), observe that by construction, the only vertices of in-degree at least two in $B$ are those in $C$ and thus, the restriction of $B$ to $V(D_M) \setminus C$ is a linear branching. 
Furthermore, since the set of leaves of $B$ is exactly $\min D_M$, we in fact have $d^-_B(u) = 1$ for every $u \in V(D_M) \setminus (C \cup \min D_M)$; we denote by $u_{in}$ the unique in-neighbor of $u$ in $B$.
Since $B$ is maximal, it follows that
\begin{equation*}
\begin{split}
\lvert U(B)\rvert &= \sum_{u \in V(D_M) \setminus C} \nu_B(u) + \sum_{u \in C} \nu_B(u)\\
&= \sum_{u \in \min D_M} \lvert  u\rvert   + \sum_{u \in V(D_M) \setminus (C \cup \min D_M)} \Bigg(\lvert  u\rvert   - \lvert  u_{in} \rvert \Bigg) + \sum_{u \in C}\Bigg(\lvert  u\rvert   - \Bigg\lvert  \bigcup_{v \in N^-_B(u)} v\Bigg\rvert \Bigg)\\
&= \sum_{u \in V(D_M)} \lvert  u\rvert - \sum_{u \in V(D_M) \setminus (C \cup \min D_M)} \lvert  u_{in}\rvert
 - \sum_{u \in C} \Bigg\lvert  \bigcup_{v \in N^-_B(u)} v \Bigg\rvert \\
&= \sum_{u \in V(D_M)} \lvert  u\rvert   - \sum_{u \in V(D_M) \setminus \max D_M\colon N^+_B(u) \cap C = \emptyset} \lvert  u\rvert  - \sum_{u \in C} \Bigg\lvert  \bigcup_{v \in N^-_B(u)} v \Bigg\rvert\\
&= \sum_{u \in \max D_M} \lvert  u\rvert + \sum_{u \in V(D_M)\colon N^+_B(u) \cap C \neq \emptyset} \lvert  u\rvert - \sum_{u \in C} \Bigg\lvert  \bigcup_{v \in N^-_B(u)} v \Bigg\rvert  \\
&=  \sum_{u \in \max D_M} \lvert  u\rvert   + \sum_{u \in X} \lvert  u\rvert   - \sum_{u \in C} \Bigg\lvert  \bigcup_{v \in N^-_B(u)} v \Bigg\rvert\,.  
\end{split}
\end{equation*}

Now let $L$ be a linear completion of $\mathcal{S}_B$. Then, by construction, $A(\mathcal{S}_B) \cup L$ is a maximal branching of $D_M$; and since the only vertices of in-degree at least two in $A(\mathcal{S}_B) \cup L$ are those in $C$, we conclude similarly that $$\lvert U(A(\mathcal{S}_B) \cup L)\rvert = \sum_{v \in X \cup \max D_M} \lvert v \rvert - \sum_{v \in C} \Bigg\lvert \bigcup_{u \in N^-_B(v)} u \Bigg\rvert.$$ Therefore, $A(\mathcal{S}_B) \cup L$ is an optimal branching of $D_M$.
\end{proof}

We remark that the above proof in fact only needs the assumption that the set of minimal elements of $D_M$ is a maximum antichain of $D_M$; uniqueness is not needed.

Next, we give an efficient algorithm for the problem of testing if a given set of pairwise independent in-stars admits a linear completion, by reducing it to the bipartite matching problem.

\smallskip
\begin{lemma}
\label{lem:polylincomp}
Let $D_M$ be an $n$-vertex containment digraph and let $\mathcal{S}$ be a set of pairwise independent in-stars of $D_M$. 
Then, deciding whether $\mathcal{S}$ admits a linear completion and computing a linear completion of $\mathcal{S}$ if there is one, can be done in time $\mathcal{O}(n^{2+o(1)}$).
\end{lemma}

\begin{proof}
We show that the problem of deciding whether $\mathcal{S}$ admits a linear completion and computing a linear completion of $\mathcal{S}$ if there is one can be reduced in time $\mathcal{O}(n^2)$ to the perfect matching problem in a bipartite graph with $\mathcal{O}(n)$ vertices.
In particular, since the perfect matching problem can be reduced in linear time to a maximum flow problem, this yields the claimed time complexity by applying the maximum flow algorithm of van den Brand et al.~\cite{maxflowlinear}.

Let $\mathcal{S} = \{(c_i,x^i_1,\ldots,x^i_{p_i})\colon i \in [\ell]\}$. 
Denoting by $C = \{c_i\colon i\in [\ell]\}$ and by $X = \{x^i_j\colon  i \in [\ell] \text{ and } j \in [p_i]\}$, let $G= (O \cup I, E)$ be the bipartite graph where
\begin{itemize}
\item $O = V(D_M) \setminus (X \cup \max D_M)$,
\item $I = V(D_M) \setminus (C \cup \min D_M)$, and
\item a vertex $u \in O$ is adjacent to a vertex $v \in I$ if and only if $(u,v) \in A(D_M)$.
\end{itemize}
Observe that the graph $G$ has $\mathcal{O}(n)$ vertices and can be computed in $\mathcal{O}(n^2)$ time. 
We show below that
$\mathcal{S}$ has a linear completion if and only if  $G$ has a perfect matching.

Assume first that $\mathcal{S}$ has a linear completion $L$ and let $M = \{uv\colon u \in O, v \in I, (u,v) \in L\}$. 
Let us show that $M$ is a perfect matching of $G$. 
Consider first a vertex $u \in O$. Since $u \notin X \cup \max D_M$, $u$ is the tail of exactly one arc $(u,v) \in L$. 
Furthermore, $v \notin \min D_M$ (since otherwise the in-degree of $v$ in $D_M$ would be zero) and $v \notin C$ (for otherwise $v$ would be the head of an arc in $L$, a contradiction to the definition of $L$). 
It follows that $v \in I$ and thus, $uv \in M$ by construction. 
Second, consider a vertex $u \in I$. Since $u \notin C \cup \min D_M$, $u$ is the head of exactly one arc $(v,u) \in L$. 
Furthermore, $v \notin \max D_M$ (since otherwise the out-degree of $v$ in $D_M$ would be zero) and $v \notin X$ (for otherwise $v$ would be the tail of an arc in $L$, a contradiction to the definition of $L$). 
It follows that $v \in O$ and thus, $vu \in M$ by construction. Therefore, $M$ is a perfect matching of $G$.

Conversely, assume that $G$ has a perfect matching $M$ and let $L = \{(u,v) \in A(D_M)\colon u \in O, v \in I, uv \in M\}$. 
Let us show that $L$ is a linear completion of $\mathcal{S}$. 
To this end, consider a vertex $u \in V(D_M)$. 
Assume first that $u \notin C \cup X$. 
If $u \notin \max D_M$, then $u \in O$ and since $M$ is a perfect matching, $u$ is matched in $M$ to a vertex $v \in I$; by construction, the arc $(u,v)$ belongs to $L$ and thus $u$ is the tail of exactly one arc in $L$. 
Similarly, if $u \notin \min D_M$, then $u \in I$ and since $M$ is a perfect matching, $u$ is matched in $M$ to a vertex $v \in O$; by construction, the arc $(v,u)$ belongs to $L$ and thus, $u$ is the head of exactly one arc in $L$. 
Assume next that $u \in C \setminus X$. 
Then $u \notin I$ and thus, $u$ is the head of no arc in $L$. 
Now if $u \in \max D_M$ then $u \notin O$ and $u$ is the head of no arc in $L$. 
Otherwise $u \in O$ and since $M$ is a perfect matching, $u$ is matched in $M$ to a unique vertex $v \in I$; by construction, the arc $(u,v)$ belongs to $L$ and thus, $u$ is the tail of exactly one arc in $L$. 
Assume next that $u \in X \setminus C$. 
Then $u \notin O$ and thus, $u$ is the tail of no arc in $L$. 
Now if $u \in \min D_M$ then $u \notin I$ and $u$ is the head of no arc in $L$. 
Otherwise, $u \in I$ and since $M$ is a perfect matching, $u$ is matched in $M$ to a unique vertex $v \in O$; by construction, the arc $(v,u)$ belongs to $L$ and thus, $u$ is the head of exactly one arc in $L$. 
Finally, if $u \in C \cap X$ then $u \notin O \cup I$ and thus, $u$ is the head and the tail of no arc in $L$. 
Therefore, $L$ is a linear completion of $\mathcal{S}$, as claimed.
\end{proof}

We now have everything ready to prove our main result.

\smallskip
\begin{theorem}
\label{thm:boundedwidth}
For any binary matrix $M\in \{0,1\}^{m \times n}$, the optimum value $\beta(D_M)$ can be computed in time 
$\mathcal{O}(n^{\wdt(D_M) \cdot (\wdt(D_M) +1)}\cdot (\wdt(D_M)^{4}+ n^{2+o(1)})+n^2m)$.
\end{theorem}

\begin{proof}
We define a weight function $w$ on $V(D_M)$ as follows: for every $u \in V(D_M)$, $w(u) = |u|$. 
We claim that the following algorithm computes $\beta(D_M)$ in the stated time.

\medskip
\noindent
\textbf{Algorithm}
\begin{itemize}
\item[1.] Compute the containment digraph $D_M$ and the weight function $w$.
\item[2.] Compute a maximum antichain $N$ of $D_M$ of maximum total weight with respect to~$w$.
\item[3.] Set $D = D_M[N \cup V_N^+]$ and initialize $\mathsf{OPT} = \infty$.
\item[4.] For every set $\mathcal{S}$ of pairwise independent in-stars of $D$ such that $|\mathcal{S}| \leq \wdt(D)$ and every in-star in $\mathcal{S}$ has size at most $\wdt(D)$ do:
\begin{itemize}
\item[4.1] If $\mathcal{S}$ admits a linear completion $L$ then $\mathsf{OPT} = \min (\mathsf{OPT}, |U(A(\mathcal{S}) \cup L)|)$.
\end{itemize}
\item[5.] Return $\mathsf{OPT}$.
\end{itemize}

\noindent
\textbf{Correctness.} Observe first that if $N = \max D_M$ then $|\max D_M| = \wdt(D_M)$ and so, $\beta(D_M) = \sum_{v \in N} |v|$ by \Cref{thrm:wdtmaxel}. 
Assume henceforth that $N \neq \max D_M$ and let us show that the above algorithm computes $\beta(D)$. 
Since by \Cref{lem:reduce}, $\beta(D_M) =\beta(D)$, this would prove the correctness of the algorithm.

Note first that the value of $\mathsf{OPT}$ is always at least $\beta(D)$. 
Indeed, if $\mathcal{S}$ is a set of pairwise independent in-stars of $D$ and $\mathcal{S}$ admits a linear completion $L$ then by definition, $A(\mathcal{S}) \cup L$ is a branching of $D$ and thus, $|U(A(\mathcal{S}) \cup L)| \geq \beta(D)$. 
Let us next show that $\mathsf{OPT}$ is at some point set to $\beta(D)$. 
By construction, $N = \min D$ is the unique maximum antichain of $D$ (indeed, if there were a maximum antichain $N' \neq N$ in $D$, then the total weight of $N'$ would be strictly larger than that of $N$, a contradiction to the definition of $N$). 
Thus, by \Cref{lem:inter}, there exists a maximal optimal branching $B$ of $D$ whose set of leaves is exactly $N$. 
 We claim that $\mathcal{S}_B = \{(c_i,x^i_1,\ldots,x^i_{p_i})\colon i \in [\ell]\}$ is a set of pairwise independent in-stars of $D$ such that
\begin{itemize}
\item[(i)] $1 \leq \, \ell \leq \wdt(D)$, and
\item[(ii)] for every $i \in [\ell]$, $p_i \leq \wdt(D)$.
\end{itemize}
Let us first show that $\ell \ge 1$.
If $\ell = 0$, then $\mathcal{S}_B = \emptyset$, that is, $B$ is a linear branching. 
Then, $|N| =  |\max D_M|$, which, by maximality of the weight of $N$, implies that $N = \max D_M$, a contradiction to our assumption.
Hence, $\ell\ge 1$, as claimed.
Next, we show that $\ell\le  \wdt(D)$.
Denote by $C = \{c_i\colon i \in [\ell]\}$ and by $X = \{x^i_j\colon  i \in [\ell] \text{ and } j \in [p_i]\}$. Since $B$ is maximal, every vertex in $V(D) \setminus \max D$ has out-degree one in $B$ and hence,
\[
\sum_{v \in V(D)} d^+_B(v) = |V(D)| - |\max D|.
\]
On the other hand, every vertex in $V(D) \setminus (C \cup \min D)$ has in-degree one by definition and so,
\begin{equation*}
\begin{split}
\sum_{v \in V(D)} d^-_B(v) &= \lvert V(D) \setminus (C \cup \min D)\rvert + \sum_{i \in [\ell]} p_i = \lvert V(D)\rvert - \left( \ell + |\min D|\right) + \sum_{i \in [\ell]} p_i \\
&= \lvert V(D)\rvert - \left(\ell + \wdt(D)\right) + \sum_{i \in [\ell]} p_i \,.
\end{split}
\end{equation*}
Since $\sum_{v \in V(D)} d^+_B(v) = \sum_{v \in V(D)} d^-_B(v)$, it follows that
\begin{equation}\label{eq}
\wdt(D) = \lvert \max D \rvert - \ell + \sum_{i \in [\ell]} p_i.
\end{equation}
Now by definition, $p_i \geq 2$ for every $i \in [\ell]$ and so, by \eqref{eq},
\[
\wdt(D) \geq \lvert \max D\rvert - \ell + 2\ell = |\max D| + \ell.
\]
Thus, item (i) holds. 
Now if there exists $i \in [\ell]$ such that $p_i > \wdt(D)$, then by \eqref{eq},
\begin{equation*}
\begin{split}
p_i &> \lvert \max D\rvert - \ell + p_i + \sum_{j \neq i} p_j
\geq \lvert \max D\rvert - \ell + p_i + 2(\ell -1)
\end{split}
\end{equation*}
and so, $\ell + \lvert \max D \rvert  < 2$, a contradiction as $\ell \geq 1$ by (i) and $\max D \neq \emptyset$. 
Thus, item (ii) holds true as well. Since any two in-stars in $\mathcal{S}_B$ are independent by construction, it follows that at some point in the run of the algorithm, the set $\mathcal{S}_B$ is considered in Step 4. 
Now by \Cref{lem:forall}(i), $\mathcal{S}_B$ admits a linear completion; and since for any linear completion $L$ of $\mathcal{S}_B$, $\lvert U(B)\rvert  = \lvert U\left(A(\mathcal{S}_B) \cup L\right) \rvert $ by \Cref{lem:forall}(ii), the variable $\mathsf{OPT}$ is set then to $\lvert U(B) \rvert = \beta(D)$ in Step 4.1.\\

\noindent
\textbf{Time complexity.} 
Step 1 can be done in time $\mathcal{O}(n^2m)$.
Next, let us show that Step 2 can be done in time $\mathcal{O}(n^{2+o(1)})$.
Let \hbox{$Q = 1 + \sum_{v \in V(D_M)} w(v)$} and let $w'$ be the weight function on $V(D_M)$ defined as follows: for every $v \in V(D_M)$, let $w'(v) = Q + w(v)$. 
Note that the weight function $w'$ can be computed in time $\mathcal{O}(n)$.

We claim that every maximum weight antichain of $D_M$ with respect to~$w'$ is a maximum antichain of $D_M$ of maximum total weight with respect to~$w$.
Indeed, let $N$ be a maximum antichain of $D_M$ of maximum total weight with respect to~$w$ and let $N'$ be a maximum weight antichain of $D_M$ with respect to~$w'$. 
Then $N'$ contains at least $\wdt(D_M)$ vertices, for otherwise 
\begin{equation*}
\begin{split}
w'(N') = \sum_{v \in N'} w'(v) \leq & (\wdt(D_M)-1)\cdot Q + \sum_{v \in N'} w(v) \\
< & (\wdt(D_M)-1)\cdot Q + Q  \leq w'(N)\,,
\end{split}
\end{equation*}
a contradiction to the maximality of $N'$.
Since any antichain in $D_M$ contains at most $\wdt(D_M)$ vertices, it follows that $|N'| = \wdt(D_M)$ and thus, $N'$ is a maximum antichain of $D_M$; in particular, \[\sum_{v \in N'} w(v) \leq \sum_{v \in N} w(v)\] by maximality of $N$. 
This implies that 
\begin{equation*}
\begin{split}
w'(N') = &\wdt(D_M) \cdot Q + \sum_{v \in N'} w(v) \\
\leq & \wdt(D_M) \cdot Q + \sum_{v \in N} w(v) = w'(N) \leq w'(N')\,,
\end{split}
\end{equation*}
and so, $w(N') =  w(N)$, as claimed. 
Now a maximum weight antichain in $D_M$ with respect to~$w'$ can be computed in time $\mathcal{O}(|A(D_M)|^{1+o(1)})$, which is in $\mathcal{O}(n^{2+o(1)})$, by \Cref{thm:max-weight-antichain}.

Step 3 can be done in time $\mathcal{O}(|V(D_M)|+|A(D_M)|) = \mathcal{O}(n^2)$.

Now each set considered in Step 4 contains at most $\wdt(D) \cdot (\wdt(D) +1)$ vertices of $D$ and so, there are at most $n^{\wdt(D) \cdot (\wdt(D) +1)}$ such sets. 
For any set $\mathcal{S}$ of in-stars of size at most $k\le \wdt(D)$, checking whether the in-stars in $\mathcal{S}$ are pairwise independent can be done in time $\mathcal{O}(k^2|\mathcal{S}|^2) = \mathcal{O}(\wdt(D)^{4})$. 
Since deciding whether a set of pairwise independent in-stars admits a linear completion and computing such a linear completion if there is one can be done in time $\mathcal{O}(n^{2+o(1)})$ by \Cref{lem:polylincomp}. Thus, Step~4 runs in time  $\mathcal{O}(n^{\wdt(D) \cdot (\wdt(D) +1)}\cdot (\wdt(D)^{4}+ n^{2+o(1)}))$.

In conclusion, the above algorithm runs in time $\mathcal{O}(n^{\wdt(D) \cdot (\wdt(D) +1)}\cdot (\wdt(D)^{4}+ n^{2+o(1)}) + n^2m)$.
Since $\wdt(D)=\wdt(D_M)$ by construction, the theorem follows.
\end{proof}

Let us remark that the above algorithm in fact computes an optimal branching of $D$ whose set of leaves is $N$. 
This branching can then be extended to an optimal branching of $D_M$ by computing an optimal linear branching of $D[N \cup V^-_N]$ using the corresponding algorithm from~\cite{PPVB} (or \Cref{thrm:wdtmaxel}).

\section{Conclusion}\label{sec:conclusion}

\begin{sloppypar}
In this paper, we studied an \textsf{NP}-hard optimization problem, the \textsc{Minimum Uncovering Branching} (MUB) problem. 
We presented two sufficient conditions for polynomial-time solvability, showing in particular that the problem is polynomial-time solvable on instances of bounded width.
This result, along with the only previously known efficiently solvable case from~\cite{CompAndAlg} dealing with the case when the out-neighborhood of each vertex of the containment digraph has width at most one, motivates the question of whether a common generalization of these two results might be possible, namely to the case when the out-neighborhood of each vertex of the containment digraph has bounded width.
However, this is not the case unless \textsf{P} = \textsf{NP}: it follows from the hardness reduction used in the proof of~\cite[Proposition 4.2]{PPVB} that the problem is \textsf{APX}-hard on instances for which each vertex of the containment digraph has width at most two.
\end{sloppypar}

Nevertheless, our work leaves open several questions and suggests several directions for future research, including the following.
\begin{enumerate}
	\item[i)] Identification of further polynomially solvable cases of the MUB problem.
	\item[ii)] Does the problem admit a constant factor approximation algorithm?
  \item[iii)] Our main result is an \textsf{XP} algorithm for the problem parameterized by the width of the binary matrix.
  Is the problem fixed-parameter tractable with respect to this parameter? 
	\item[iv)] Studying further extensions of the model that could be particularly relevant for the biological application, for instance, when the input binary matrix has partially missing data or the data may contain errors.
\end{enumerate}

\paragraph{Acknowledgments.}
The authors are grateful to the anonymous reviewers for their valuable suggestions and to Ekkehard K\"ohler for helpful discussions, in particular, for pointing out reference~\cite{Mohring}.
This work is supported in part by the Slovenian Research and Innovation Agency (I0-0035, research program P1-0285 and research projects J1-3003, J1-4008, J1-4084, J1-60012, and N1-0370) and by the research program CogniCom (0013103) at the University of Primorska.

\bibliographystyle{plain}
\bibliography{bibliography}

\end{document}